\documentclass[11pt,twoside]{article}

\usepackage{parskip}
\usepackage{fullpage}
\usepackage{comment}
\usepackage{epsf}
\usepackage{fancyheadings}
\usepackage{graphics}
\usepackage{graphicx}
\usepackage{psfrag}
\usepackage{color}
\usepackage{bbold}
\usepackage{amsthm}
\usepackage{amsfonts}
\usepackage{amsmath}
\usepackage{amssymb}
\usepackage{graphicx,epsfig,graphics,epsf}
\usepackage{float,amsmath,amsfonts,amssymb,amsthm}
\usepackage{mathtools}
\usepackage{subcaption}
\usepackage[noend]{algorithmic}
\usepackage[ruled,vlined]{algorithm2e}
\usepackage{url}
\usepackage{fullpage}
\usepackage{makeidx}
\usepackage{enumerate}

\usepackage{graphicx,float,psfrag,epsfig,caption}
\usepackage[usenames,dvipsnames,svgnames,table]{xcolor}

\usepackage{sansmathfonts}
\theoremstyle{plain}

\long\def\comment#1{}

\newcommand{\Prob}{\ensuremath{{\mathbb{P}}}}

\newtheorem{claim}{Claim}[section]
\newtheorem{lemma}[claim]{Lemma}

\newtheorem{conjecture}[claim]{Conjecture}
\newtheorem{theorem}{Theorem}

\newtheorem{definition}[claim]{Definition}

\theoremstyle{myremark}

\usepackage{tikz}
\usetikzlibrary{shapes, arrows, calc, positioning,matrix,math}
\tikzset{
data/.style={circle, draw, text centered, minimum height=3em ,minimum width = .5em, inner sep = 2pt},
empty/.style={circle, text centered, minimum height=3em ,minimum width = .5em, inner sep = 2pt},
}
\usepackage{pgfplots}
\pgfplotsset{compat=1.5}
\usepgfplotslibrary{fillbetween}
\usetikzlibrary{patterns}

\newcommand{\G}{{\sf G}}

\newcommand{\A}{\mathcal{A}}

\newcommand{\poly}{\operatorname*{poly}}

\usepackage{bbm}
\usepackage[inline]{enumitem}
\usepackage{enumitem}

\usepackage{tabularx,ragged2e}

\usepackage[pagebackref,letterpaper=true,colorlinks=true,pdfpagemode=none,citecolor=OliveGreen,linkcolor=BrickRed,urlcolor=BrickRed,pdfstartview=FitH]{hyperref}

\title{Logspace Reducibility From Secret Leakage Planted Clique}
\author{Jay Mardia\thanks{Department of Electrical Engineering, Stanford University. jmardia@stanford.edu}}
\date{}
\begin{document}

\maketitle
\begin{abstract}
The planted clique problem is well-studied in the context of observing, explaining, and predicting interesting computational phenomena associated with statistical problems. When equating computational efficiency with the existence of polynomial time algorithms, the computational hardness of (some variant of) the planted clique problem can be used to infer the computational hardness of a host of other statistical problems.
\begin{quote}
	\textit{Is this ability to transfer computational hardness from (some variant of) the planted clique problem to other statistical problems robust to changing our notion of computational efficiency to space efficiency?}
\end{quote}
We answer this question affirmatively for three different statistical problems, namely Sparse PCA \cite{berthet2013complexity}, submatrix detection \cite{ma2015computational}, and testing almost $k$-wise independence \cite{alon2007testing}. The key challenge is that space efficient randomized reductions need to repeatedly access the randomness they use. Known reductions to these problems are all randomized and need polynomially many random bits to implement. Since we can not store polynomially many random bits in memory, it is unclear how to implement these existing reductions space efficiently. There are two ideas involved in circumventing this issue and implementing known reductions to these problems space efficiently.
\begin{enumerate}
\item When solving statistical problems, we can use parts of the input itself as randomness. This idea was pioneered by \cite{goldreich2002derandomization} in the context of derandomizing `typically correct' algorithms. Our observation is that this is useful even if we do not care about derandomization, and instead care about repeatedly accessing randomness that can not be stored in working memory.
\item Secret leakage variants of the planted clique problem with appropriate secret leakage can be more useful than the standard planted clique problem when we want to use parts of the input as randomness. The idea that secret leakage is helpful in the context of polynomial time reducibility to statistical problems was introduced by \cite{pmlr-v125-brennan20a}. They demonstrated the polynomial time hardness of a wide variety of statistical problems by assuming polynomial time hardness of secret leakage variants of the planted clique problem. Our observation is that leakage is useful even for the seemingly unrelated task of finding randomness in the input.
\end{enumerate}
\end{abstract}

\setcounter{tocdepth}{2}

\section{Introduction}
\label{sec:Introduction}

Over the past century, the field of statistics has developed a rich theory to explain the ease or difficulty of various inference problems. In parallel, computational complexity theorists have also developed a different rich theory to explain the ease or difficulty of various computational problems. It has long been observed \cite{valiant1984theory,decatur2000computational,servedio1999computational} that these notions governing statistical and computational tractability of a problem do not always coincide. In recent decades this misalignment between the two notions of tractability (called a statistical-computational gap) has emerged as a widespread conjectural phenomenon in high-dimensional statistics and learning theory. Even well-studied problems where no such gap exists might display such a gap with minor tweaks. For example, this happens when the statistician is simply given additional information that they are unable to efficiently exploit \cite{berthet2013complexity,berthet2013optimal,wang2016statistical}, or is required to handle small amounts of data contamination \cite[Robust sparse mean estimation]{li2017robust}. Hence, there has been much interest and success in developing tools to understand and predict the occurrence of such gaps. The introduction of \cite{pmlr-v125-brennan20a} has an excellent survey on the current landscape of statistical-computational phenomena.

Most research on the computational complexity of statistical inference has used the existence of polynomial time algorithms as a benchmark for computational tractability. This is an excellent model for computational efficiency whose use in the statistical context has been validated by it enabling the discovery and study of statistical-computational gaps. That being said, traditional complexity theory has identified a host of resources such as memory, communication, and randomness that capture the effect of some fundamental resource  other than time. Studying the complexity theories of these resources is not only interesting in their own right, but has produced insights in seemingly unrelated areas. For example, space complexity (one way of formalizing limited memory) turns out to have fundamental connections to the complexity of 2-player games \cite[The essence of PSPACE]{arora2009computational} as well as that of interactive proofs \cite{shamir1992ip}. Communication complexity turns out to be very useful in studying streaming complexity (another way of formalizing limited memory).

Even in the context of statistical problems, studying alternate notions of resource has proved fruitful. Evidence for the computational hardness of several problems has been built up by showing that various restricted algorithmic classes can not solve these problems (see Section~\ref{sec:related-work} for details). While these classes all correspond to some non-standard notion of resource, to the best of our knowledge, they are usually viewed as proxies for the class of ``polynomial time algorithms''. That is, these results are usually used to draw heuristic conclusions about the entire class of polynomial time algorithms. One way to interpret this state of affairs is that in the context of statistical problems, studying alternate notions of complexity is not merely interesting; it is also a promising avenue towards results that are actually within reach of our current mathematical technology.

The preceding discussion recommends, in the statistical context, an explicit study of resource-driven notions of efficient computation other than ``polynomial time''\footnote{In this spirit, \cite{moitra2020parallels} studied connections between small depth circuit classes and phase transitions in the broadcast tree model and \cite{rashtchian2021average} studied the communication complexity of statistical problems.}. A natural question that arises is-
\begin{quote}\centering
	\textit{Are the conjectured statistical-computational phenomena we know robust to changes in our notion of computational efficiency?}
\end{quote}
While this is a very broad question, in this work we will be concerned with one specific alternate notion of efficiency (space efficiency) and the statistical-computational phenomena centered around the planted clique problem\footnote{This problem concerns an $n$-vertex graph with a statistical signal whose strength is denoted by an integer $1 \leq k \leq n$.}. Space is one of the most well-studied computational resources other than time while the planted clique problem has emerged as a testbed for the study of statistical-computational gaps. When equating computational efficiency with the existence of polynomial time algorithms, this problem displays the following interesting statistical-computational phenomena.
\begin{enumerate}[leftmargin=0pt]
\item There is a conjectured threshold phenomenon at signal strength $k = \Theta(\sqrt{n})$, irrespective of whether the problem is formalized using its detection or recovery variant.
\item The computational hardness of (some variant of) the planted clique problem below the $\Theta(\sqrt{n})$ signal strength threshold can be used to infer the computational hardness of a host of other statistical problems.
\end{enumerate}
Recently \cite{mardia2020space} provided an $O(\log^{*}n \cdot \log n)$ space algorithm for planted clique recovery when $k = \Theta(\sqrt{n})$ (and it is easy to see that a logspace detection algorithm exists). This is some evidence that the threshold phenomena at $k = \Theta(\sqrt{n})$ also occurs at the same threshold when we use space efficiency as our notion of computational efficiency.
\begin{quote}
\textit{In this work, we investigate whether the ability to transfer computational hardness from (some variant of) the planted clique problem to other statistical problems is also robust to changing our notion of computational efficiency to space efficiency.}
\end{quote}

The conjectured non-existence of polynomial time algorithms\footnote{We do not discuss the exact model of polynomial time algorithms used by previous works. However, a crucial aspect is that the model allows the use of randomness.} for the planted clique problem and its variants below the $k=\Theta(\sqrt{n})$ signal strength threshold has been used to show non-existence of polynomial time algorithms for various other statistical problems. Our goal is to replace the phrase `polynomial time algorithms' in the sentence above with `randomized logspace algorithms', which is the correct notion of space efficent algorithms. Section~\ref{subsec:obstacle} provides details about this class of algorithms, but for now we just need to recall that any randomized logspace algorithm is also a polynomial time algorithm. As a result, this work \textit{will not} provide any \textit{new} hardness beliefs. That is, we already believe these other statistical problems do not have polynomial time algorithms, so we also already believe they do not have randomized logspace algorithms. Our motivation is, instead, the following.
\begin{enumerate}
\item As discussed, our main motivation is to examine the robustness of statistical-computational phenomena (particularly, the ability to transfer hardness from one statistical problem to another) to changing our notion of computational efficiency from time bounded to space bounded.
\item The most natural route towards proving unconditional computational hardness results against the class of polynomial time algorithms is by showing hardness against larger and larger restricted classes of algorithms. It is plausible that technology showing unconditional impossibility results for the planted clique problem and its variants for randomized logspace algorithms will appear long before corresponding technology for the entire class of polynomial time algorithms. In such a scenario, our results will imply unconditional randomized logspace hardness for various other statistical problems while existing randomized polynomial time reductions would not.
\item In worst-case complexity, many complexity classes are known to have natural complete problems that are complete under reductions far more efficient than ``polynomial time'' \cite{immerman2012descriptive,allender1997first}. For example, NP has many natural complete problems under ``first-order projections'' (fops), a type of reduction even more restrictive than logspace computation \cite{dahlhaus1983reduction,medina1994syntactic}. Similarly L, NL, and P also have problems complete under fops reductions \cite[Chapter 3]{immerman2012descriptive}, and PSPACE has problems complete under logspace reductions \cite{arratia2003note}.

The planted clique problem, with its variants, is emerging as a useful central problem for the complexity of statistical problems, analogous to central problems in worst-case complexity. Showing that it shares other characteristics of important problems in worst-case complexity, such as reducibility using reductions much more efficient than polynomial time, would strengthen this correspondence.

\item Another reason to implement known reductions space efficiently is that even if the conjectured polynomial time hardness of the planted clique problem turns out to be false, it is possible that the problem is still logspace hard. Some of the strongest evidence for the polynomial time hardness of the planted clique problem comes from the failure of low-degree polynomials and SoS convex relaxations for this problem. However, \cite{zadik2021lattice} (building on \cite{pmlr-v65-andoni17a,song2021cryptographic}) showed that the LLL algorithm implies polynomial time algorithms for some problems even though low-degree polynomials and SoS convex relaxations fail. Since it is unlikely that the LLL algorithm can be implemented space efficiently, these problems may still be logspace hard.
\end{enumerate}

\section{Our Results and Techniques}
\label{sec:results-and-techniques}
Our results suggest that the ability to transfer computational hardness from (some variant of) the planted clique problem to other statistical problems is indeed robust to changing our notion of computational efficiency to space efficiency. Informally, these results are as follows.

\noindent\fbox{%
	\parbox{\textwidth}{%
		The non-existence of randomized logspace algorithms solving (some variant of) the planted clique problem implies the non-existence of randomized logspace algorithms solving Sparse PCA hypothesis testing \cite{berthet2013complexity}, submatrix detection \cite{ma2015computational}, and testing almost $k$-wise independence \cite{alon2007testing}.
	}%
}

We spend the rest of Section~\ref{sec:results-and-techniques} explaining the statement above. We do not formally describe these other statistical problems\footnote{The formulation of Sparse PCA used by \cite{berthet2013complexity} (and by us) is just one of many formulations considered in the literature. There has been a significant amount of work \cite{berthet2013optimal,wang2016statistical,gao2017sparse,pmlr-v75-brennan18a,brennan2019optimal} on showing hardness of more canonical formulations of Sparse PCA. In this work, however, we restrict ourselves to the \cite{berthet2013complexity} reduction and refer the reader to \cite{brennan2019optimal} for a comprehensive discussion of issues surrounding these various formulations.} (or the known reductions in \cite{berthet2013complexity,ma2015computational,alon2007testing}) in this technical overview, since their details are not needed for our current discussion. The only thing the reader needs to know is that these known reductions all seem to require $\poly(n)$ bits of randomness to implement.

The planted clique detection problem ${\sf{PC_D}}(n,k)$ (Definition~\ref{def:pc-detection}) is a hypothesis testing problem where the input is either an Erd\H{o}s-R\'enyi graph $\G(n, 1/2)$ (a uniformly random $n$-vertex graph) or a planted clique graph $\G(n,1/2,k)$ ($\G(n, 1/2)$ with an additional clique on a uniformly random $k$-vertex subset). We focus on $k=O(n^{\frac{1}{2}-\delta})$ for any constant $\delta >0$, where efficient algorithms for the planted clique problem are believed not to exist.

Some of our results will use the hardness of variants of the planted clique problem that we introduce later. When designing these variants, we need to ensure that they still remain computationally hard in the appropriate parameter regime. To do this, we follow principles used in \cite{pmlr-v125-brennan20a}. We want the distribution of possible planted clique locations (which is no longer uniform) to have `high enough entropy' (so any potential algorithm still does not know much about the clique location), and be `symmetric enough'. The latter requirement is because evidence for hardness of the variants used in \cite{pmlr-v125-brennan20a} comes from the failure of low-degree polynomials. \cite{holmgren2021counterexamples} showed that inferring hardness from the failure of low-degree polynomials can be misleading if the problem instances are not `symmetric enough'. We discuss evidence for the hardness of the variants we use in Section~\ref{subsec:security-of-cl}.

Let $\mathcal{P}$ denote any hypothesis testing problem (Definition~\ref{defn:hyp-test}) such that there exist distributions $P_0$ and $P_1$ (not necessarily in the null or alternate hypotheses of $\mathcal{P}$) on $\poly(n)$ sized bit strings with the following properties.
\begin{enumerate}
\item There is a polynomial time algorithm using $\poly(n)$ independent uniform random bits that maps a sample drawn according to $\G(n, 1/2)$ (respectively $\G(n,1/2,k)$) to a sample drawn according to $P_0$ (respectively $P_1$).
\item $P_0$ (respectively $P_1$) is close in total variation distance to some distribution in the null (respectively alternate) hypothesis of $\mathcal{P}$.
\end{enumerate}
\cite{berthet2013complexity}, \cite{ma2015computational}, and \cite{alon2007testing} respectively showed that the hypothesis testing problems Sparse PCA, submatrix detection, and testing almost $k$-wise independence all satisfy such conditions. The non-existence of polynomial time algorithms solving the planted clique problem then immediately implies the non-existence of polynomial time algorithms solving such a hypothesis testing problem $\mathcal{P}$.

\subsection{An obstacle to space efficient randomized reductions: access to randomness}
\label{subsec:obstacle}
The class of logarithmic space bounded computation is widely agreed to be the benchmark of ‘space efficient’ computation. A deterministic or randomized logspace algorithm\footnote{Our algorithms are defined using any fixed reasonable Random-Access Turing Machine model.} (see \cite{arora2009computational} or Definition~\ref{defn:bdd-space-alg}) has read-only multiple access to an $\poly(n)$-bit input, along with read-and-write multiple access to a working tape consisting of $O(\log n)$ bits. Its output must be written on a write-only tape and the algorithm must terminate within $\poly(n)$ time\footnote{The section on $BP_HSPACE$ in \cite{saks1996randomization} and the section on Finite Automata and Counting in \cite{wigderson2019mathematics} discuss why a reasonable model for randomized space bounded computation should also have an associated time bound.}. Further, if the algorithm is randomized, it is allowed access to a \textbf{read-once} bit string of independent uniform random bits of size $\poly(n)$.

One aspect of randomized logspace algorithms worth keeping in mind (because it is the central nuisance we deal with in this work) is that the algorithm can access its random bits only in a read-once fashion. That is, if the algorithm wants to look at a random bit it has already used, it must have stored that bit for the future using its $O(\log n)$ bits of working space.

It will also be convenient to talk about ``multiple access randomized logspace algorithms'' (Definition~\ref{defn:multacc-bdd-space-alg}). These are randomized logspace algorithms given the extra ability to access not only their input and work tapes, but also polynomially many auxiliary random bits as many times as desired without having to explicitly store them in memory.

Of course, while this is a useful definition, this class of algorithms \textit{does not} correspond to efficient space bounded computation \cite{nisan1993read}, since it lets us store polynomially many random bits `for free'. Another way to see which formulation should correspond to efficient space bounded randomized computation is to see which definition allows results analogous to those known for efficient time bounded randomized computation. For example, allowing only read-once access to randomness lets us easily and unconditionally obtain results like RL $\subseteq$ NL, analogous to RP $\subseteq$ NP. See \cite{nisan1993read} for further discussion about this issue.

The reductions in \cite{berthet2013complexity,ma2015computational,alon2007testing} as well as the transference of computational hardness from the planted clique problem are still valid if we replace `polynomial time algorithms' with `multiple access randomized logspace algorithms' in the discussion above. This follows by examining the reductions implemented by these works and verifying that all the steps can be performed space efficiently when given multiple access randomness\footnote{The only additional primitive we need is the ability to approximately implement a uniformly random permutation using a multiple access randomized logspace algorithm. This is not hard, and we provide an algorithm in Lemma~\ref{lem:permutation}.}. For completeness, we provide the straightforward formal proofs of these facts in Section~\ref{sec:reductions-to-other-probs}. 

Unfortunately, such a reduction \textit{does not} let us transfer randomized logspace hardness from the planted clique problem (or its variants) to other statistical problems. The existence of a randomized logspace algorithm that solves the various other statistical problems would only imply the existence of a multiple access randomized logspace algorithm that solves some variant of the planted clique problem. This would not contradict the conjectured non-existence of a randomized logspace algorithm solving some variant of the planted clique problem\footnote{If we assume that no multiple access randomized logspace algorithm (or the even stronger non-uniform model of computation concerning polynomial sized branching programs) can solve the planted clique problem or its variants, we would, of course, be done. This strong assumption (just like the strong assumption about non-existence of polynomial time algorithms) would imply the randomized logspace hardness of our other statistical problems of interest. However, our aim in this work, as motivated in Section~\ref{sec:Introduction}, is to show that such a result also follows from the much weaker assumption concerning only efficient space bounded computation. We want to show that randomized logspace hardness can indeed be transferred from the planted clique problem or its variants to other statistical problems.}.

The reason we need to resort to multiple access randomized logspace reductions when trying to implement a randomized reduction space efficiently is as follows. A sample drawn from the distribution we are reducing to ($P_0$ or $P_1$) has size $\poly(n)$, so unlike polynomial time algorithms, a logspace algorithm of any variety can not compute and simultaneously store every bit of this sample in its memory. What is feasible, however, is that any individual bit of such a sample can be computed space efficiently given the planted clique input and randomness. A space efficient algorithm $\mathcal{A}$ that solves the hypothesis testing problem $\mathcal{P}$ can be used to space efficiently solve the planted clique problem by recomputing bits of the sample we have reduced to whenever the algorithm $\mathcal{A}$ needs them. As a logspace algorithm, $\mathcal{A}$ can inspect any bit of its input sample many times. The reduction from the planted clique problem thus needs to be able to recompute its output bits multiple times. This is not possible if the randomness it uses is only accessible once. This constitutes a barrier to implementing randomized reductions space efficiently. In fact, such a barrier has been encountered even in worst-case complexity, necessitating the use of multiple access randomized logspace reductions \cite{cai1996existence,van1998deterministic}. \cite{van1998deterministic} even states it as an open problem to implement such reductions in the more natural read-once model of randomized logspace algorithms.

\subsection{Can we use pseudorandom generators that fool space bounded computation?}

Since we only want to run space efficient reductions, a natural idea to obtain multiple access randomness for our purposes is to use a pseudorandom generator that fools space bounded computation.

\begin{enumerate}[leftmargin=0pt]
\item \textbf{Unconditional PRGs:} One might hope that the unconditional pseudorandom generators from \cite{babai1992multiparty,nisan1992pseudorandom,impagliazzo1994pseudorandomness}, combined with an idea from  Section~\ref{subsec:main-basic-pc-harvesting} to obtain the seed, would let us implement all our desired reductions.

This natural idea fails for the simple reason that in our applications, we must use these pseudorandom bits to simulate the randomness of a multiple access randomized logspace algorithm rather than a randomized logspace algorithm. All the known unconditional pseudorandom generators which can be implemented in logspace and which fool space bounded computation \cite{babai1992multiparty,nisan1992pseudorandom,impagliazzo1994pseudorandomness} need the algorithms they fool to only access their random bits in a read-once fashion. They do not fool multiple access randomized logspace algorithms and are not useful for us.

\item \textbf{Conditional PRGs which imply that every multiple access randomized logspace algorithm can be simulated by a deterministic logspace algorithm:} \cite{klivans2002graph} provides a pseudorandom generator that can be implemented in logspace and fools polynomial sized circuit families, and hence multiple access randomized logspace algorithms. In particular, this PRG implies that every multiple access randomized logspace algorithm can be simulated using a deterministic logspace algorithm. This would be enough for us to implement the reductions we consider in this work.

However, the \cite{klivans2002graph} pseudorandom generator is based on unproven assumptions. It would let us \textit{conditionally} transfer randomized logspace hardness from the planted clique problem to Sparse PCA \cite{berthet2013complexity}, submatrix detection \cite{ma2015computational}, and testing almost $k$-wise independence \cite{alon2007testing}. While this gives us some reason to believe that our desired reductions can be carried out space efficiently, this \textit{conditionality} would be less than satisfactory for the following reasons.\begin{enumerate}		
		\item We can already show the impossibility of space efficient algorithms solving the statistical problems from \cite{berthet2013complexity,ma2015computational,alon2007testing} by assuming the impossibility of time efficient algorithms solving some other problem (like the planted clique problem or its variants). Our goal is to derive such conclusions while only assuming randomized logspace hardness of some initial problem. Using the \cite{klivans2002graph} pseudorandom generator does not achieve this goal since it assumes impossibility results against `circuit size' (and hence time).
		\item One benefit of the efficient unconditional reductions we will implement arises when technology to prove impossibility results has progressed to showing unconditional impossibility results for randomized logspace algorithms solving the planted clique problem or its variants. Our results will then imply unconditional randomized logspace hardness of the statistical problems in \cite{berthet2013complexity,ma2015computational,alon2007testing}. If we used the conditional pseudorandom generator of \cite{klivans2002graph}, this would not be true.\end{enumerate}
\end{enumerate}

\noindent\fbox{%
	\parbox{\textwidth}{%
In this paper we will see that very easy to implement ideas can let us \textit{unconditionally} transfer randomized logspace hardness to various \textit{statistical problems} without making any computational assumptions beyond the \textbf{randomized logspace hardness} of the planted clique problem and its variants.
	}%
}

\subsection{Harvesting multiple access randomness from the input: hardness of Sparse PCA for a restricted range of parameters}
\label{subsec:main-basic-pc-harvesting}
The first of the two themes in this work is the following.
\begin{quote}
\textit{There is an algorithmic technique available in the context of statistical problems (but not for worst-case problems) which might enable space efficient randomized reductions using randomized logspace algorithms that do not have their own supply of auxiliary multiple access randomness. We can harvest multiple access randomness from the input itself.}
\end{quote}

\cite{goldreich2002derandomization} pioneered the idea that randomized algorithms which are ``typically correct'' can be derandomized by harvesting randomness from the input itself. This idea has often been used for various kinds of derandomization tasks (see the surveys \cite{shaltiel2010typically,hemaspaandra2012sigact} or the related work section in \cite{hoza2017typically}). Our observation is that harvesting randomness from average case / statistical inputs has a use beyond derandomization. Namely, to simulate multiple access to randomness in situations where we otherwise only have read-once randomness.

We show that when solving statistical problems with nice structure such as the planted clique problem and its variants, one can simulate a multiple access randomized logspace algorithm using a randomized logspace algorithm. The non-existence of the latter then implies the non-existence of the former, which, as discussed earlier, implies the non-existence of randomized logspace algorithms for other statistical problems such as Sparse PCA, submatrix detection, and testing almost $k$-wise independence.
\begin{quote}
(Lemma~\ref{lem:basic-randomness-harvesting}, Informal) If no randomized logspace algorithm can solve the planted clique problem ${\sf{PC_D}}(n,k)$, no multiple access randomized logspace algorithm using at most ${n \choose 2} - {n-m \choose 2} = o(n^2 / k)$ multiple access random bits can solve ${\sf{PC_D}}(n-m,k)$ for `nice enough'\footnote{This niceness involves a very mild technical assumption.} values of $m=o(n/k)$.
\end{quote}

We sketch the the proof of Lemma~\ref{lem:basic-randomness-harvesting}. Given an instance of the planted clique detection problem ${\sf{PC_D}}(n,k)$, we can harvest $\Theta(n \cdot m)$ independent multiple access random bits to use as long as $m = o(n/k)$. To do this, use the edge indicators for the ${n \choose 2}-{n-m \choose 2} = \Theta(n \cdot m)$ possible edges involving at least one of the last $m$ vertices as multiple access randomness. Then view the subgraph induced on the first $n-m$ vertices as an instance of ${\sf{PC_D}}(n-m,k)$. With high probability there will be no planted clique vertex in the last $m= o(n/k)$ vertices, so both our `multiple access randomness' and our `planted clique instance' on the first $n-m$ vertices will behave as desired. The claims of Lemma~\ref{lem:basic-randomness-harvesting} then follow from a simple contradiction argument. See Figure~\ref{fig:basic-pc} for a  pictorial representation.

\tikzmath{\coord = 3;\random = 0.07 * \coord;\reduc = 0.37*\coord;\reducrand = 0.39*\coord;}
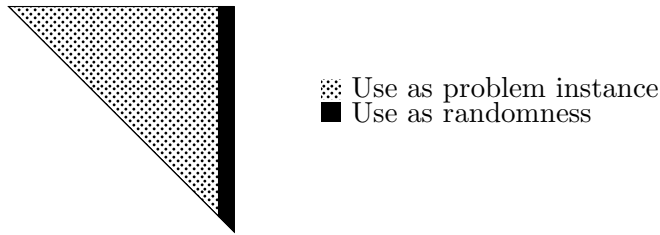
\begin{figure}[h]
	\begin{center}
	\begin{tikzpicture}
	\draw  (\coord,\coord) to (\coord,0) to (0,\coord) to (\coord,\coord);
	\draw [fill] (\coord,\coord) to (\coord-\random,\coord) to (\coord-\random,\random) to (\coord,0) to (\coord,\coord);
	\draw [pattern = crosshatch dots] (\coord-\random,\coord) to (0,\coord) to (\coord-\random,\random) to (\coord-\random,\coord) ;
	\end{tikzpicture}
	\begin{tikzpicture}
	\matrix {
		\node at (0,0) (D) {};
		\node at (0.3*\coord,0.55*\coord) [pattern = crosshatch dots,label=right: Use as problem instance] {};
		\node at (0.3*\coord,0.45*\coord) [fill,label=right: Use as randomness] {};\\
	};
	\end{tikzpicture}
	\caption{Adjacency matrix of the input graph above the diagonal as used in Lemma~\ref{lem:basic-randomness-harvesting}.}
	\label{fig:basic-pc}
\end{center}
\end{figure}

One drawback of the scheme from Lemma~\ref{lem:basic-randomness-harvesting} is that it can only rule out multiple access randomized logspace algorithms solving ${\sf{PC_D}}(n-m,k)$ if they use $o(n^2/k)$ independent multiple access random bits. In principle, however, a multiple access randomized logspace algorithm solving ${\sf{PC_D}}(n-m,k)$ can use any arbitrarily large $\poly(n)$ number of independent multiple access random bits.

The Sparse PCA problem \cite{berthet2013complexity} has an input consisting of $\bar{n}$ vectors in $\bar{d}$ dimensions. The reduction in Lemma~\ref{lem: spca-br13-redn} shows that when $\bar{d}$ is not much larger than $\bar{n}$, a randomized logspace algorithm solving Sparse PCA implies a multiple access randomized logspace algorithm solving ${\sf{PC_D}}(n-m,k)$ whose multiple access randomness usage contradicts the logspace planted clique conjecture by Lemma~\ref{lem:basic-randomness-harvesting}. For this reasonable but restricted parameter range, this lets us deduce non-existence of randomized logspace algorithms solving Sparse PCA from the conjectured non-existence of randomized logspace algorithms for the planted clique problem (Conjecture~\ref{conj:logspace-pc}). See Theorem~\ref{thm:spca-restricted} in Section~\ref{subsec:spca} for a formal statement and details about the parameters involved.

\textbf{The remaining challenge:} Reductions to Sparse PCA for more general values of $\bar{d}$, to submatrix detection, and to testing almost $k$-wise independence can not be carried out with the amount of multiple access randomness provided by Lemma~\ref{lem:basic-randomness-harvesting}. For example, when reducing from ${\sf PC_D}(n-m,k)$ to submatrix detection, we need $\widetilde{\Theta}(n^2)$ multiple access random bits \cite{ma2015computational}. We need some further ideas to make progress.

\subsection{Self-reducibility and $k$-partite planted clique: hardness of submatrix detection for a restricted range of parameters}

In Section~\ref{subsec:main-basic-pc-harvesting} we saw that within a planted clique instance we can find two things; another slightly smaller planted clique instance as well as independent random bits that are disjoint from this smaller instance. These multiple access random bits can be used to simulate a multiple access randomized logspace algorithm that is run on the smaller planted clique instance. Unfortunately, the number of multiple access random bits we get from this scheme is not as large as we might want.

Our next observation is that to reduce to other statistical problems, we do not care solely about the number of multiple access random bits we can harvest from the original planted clique instance. Instead, we care about this number \textit{in comparison to} the size of the smaller planted clique instance, since this smaller instance is what we reduce from. We can thus make progress not just by harvesting more multiple access randomness from the input, but also by locating a smaller instance of the planted clique problem in the input!

The planted clique problem almost has a useful `self-reducibility' property. Given an instance of ${\sf PC_D}(n,k)$, the subgraph induced on the first $n_s$ vertices is morally like an instance of ${\sf PC_D}(n_s,k \cdot (n_s/n))$. We can use the same randomness harvesting technique from Section~\ref{subsec:main-basic-pc-harvesting} to get multiple access to slightly less than $n^2/k$ random bits. See Figure~\ref{fig:basic-kpc} for a pictorial representation.

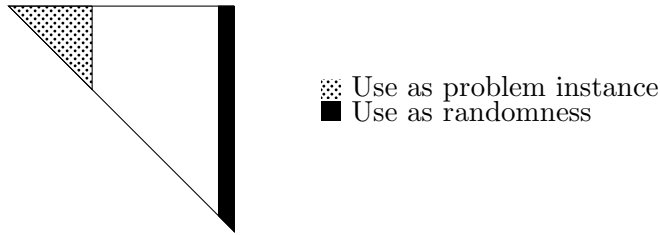
\begin{figure}[h]
	\begin{center}
		\begin{tikzpicture}
\draw  (\coord,\coord) to (\coord,0) to (0,\coord) to (\coord,\coord);
\draw [fill] (\coord,\coord) to (\coord-\random,\coord) to (\coord-\random,\random) to (\coord,0) to (\coord,\coord);
\draw [pattern = crosshatch dots] (\reduc,\coord) to (0,\coord) to (\reduc,\coord-\reduc) to (\reduc,\coord) ;
\end{tikzpicture}
\begin{tikzpicture}
		\matrix {
			\node at (0,0) (D) {};
			\node at (0.3*\coord,0.55*\coord) [pattern = crosshatch dots,label=right: Use as problem instance] {};
			\node at (0.3*\coord,0.45*\coord) [fill,label=right: Use as randomness] {};\\
		};
		\end{tikzpicture}
		\caption{Adjacency matrix of the input graph above the diagonal as used in Lemma~\ref{lem:basic-kpc-randomness-harvesting}.}
		\label{fig:basic-kpc}
	\end{center}
\end{figure}

Unfortunately, the subgraph induced on the first $n_s$ vertices is not actually an instance of ${\sf PC_D}(n_s,k \cdot (n_s/n))$. This is because the size of the planted clique in this subgraph is a binomial random variable with mean $k \cdot (n_s/n)$. With probability tending to $1$, this size will in fact \textit{not} be $k \cdot (n_s/n)$. However, several close variants of the planted clique problem do have an exact self-reducibility property. One such variant is where every vertex of the graph is included in the planted clique independently with probability $k/n$. It is easy to see that the subgraph induced on any subset of vertices of such a graph will also exactly be distributed like such a graph ensemble with a different signal strength\footnote{The self-reducibility property of this variant was crucially used in \cite{mardia2020finding} to study the sublinear time complexity of the planted clique problem.}.

Another variant called the $k$-partite planted clique problem ${\sf{kPC_D}}(\ell \cdot k, k)$ (Definition~\ref{def:kpc-detection}) studied in \cite{pmlr-v125-brennan20a} also has such an exact self-reducibility property. In this variant, the number of vertices in the graph is an integer multiple of the number of the planted clique vertices. This is why we prefer to parametrize the problem as ${\sf{kPC_D}}(\ell \cdot k, k)$, denoting a graph with $\ell \cdot k$ vertices and a planted clique of size $k$. The difference between this formulation and the vanilla planted clique problem ${\sf{PC_D}}(\ell \cdot k, k)$ is that in the planted graph, every set of $\ell$ vertices of the form $\{(i-1)\cdot \ell + 1,...,i \cdot \ell\}$ (for $i \in \{1,2,...,k\}$) has exactly one planted clique vertex, which is distributed uniformly over this set. Given an instance of ${\sf{kPC_D}}(\ell \cdot k, k)$, the subgraph induced on the first $\ell \cdot k_s$ vertices is exactly an instance of ${\sf{kPC_D}}(\ell \cdot k_s, k_s)$.

This problem ${\sf{kPC_D}}(\ell \cdot k, k)$ is similar enough to the planted clique problem that it is still believed to be unsolvable by polynomial time (and hence randomized logspace) algorithms when $k = O(\ell^{1-\delta})$ for some constant $\delta > 0$. \cite{pmlr-v125-brennan20a} provide evidence for this belief through impossibility results for low degree algorithms. This means the scheme depicted in Figure~\ref{fig:basic-kpc} can be implemented for the $k$-partite planted clique problem.
\begin{quote}
(Lemma~\ref{lem:basic-kpc-randomness-harvesting}, Informal) If no randomized logspace algorithm can solve the $k$-partite planted clique problem ${\sf{kPC_D}}(\ell \cdot k,k)$, no multiple access randomized logspace algorithm using at most ${\ell \cdot k \choose 2} - {\ell \cdot k-m \choose 2} = o(\ell^2 \cdot k)$ multiple access random bits can solve ${\sf{kPC_D}}(\ell \cdot k_s,k_s)$ for `nice enough' values of $m=o(\ell)$ and $k_s \leq k-1$.
\end{quote}

The \cite{ma2015computational} reduction from a $k$-partite planted clique instance on $\ell \cdot k_s$ vertices to submatrix detection requires $\widetilde{\Theta}((\ell \cdot k_s)^2)$ multiple access random bits (Lemma~\ref{lem:submat-redn}). Hence, the existence of a randomized logspace algorithm solving submatrix detection yields a multiple access randomized logspace solving ${\sf{kPC_D}}(\ell \cdot k_s,k_s)$ with $\widetilde{\Theta}((\ell \cdot k_s)^2)$ multiple access random bits. This only contradicts Lemma~\ref{lem:basic-kpc-randomness-harvesting} if $k_s$ is small enough. Even if we take $k$ in Lemma~\ref{lem:basic-kpc-randomness-harvesting} as only slightly smaller than $\ell$, we need $k_s$ to be smaller than $o(\sqrt{\ell})$ for a contradiction to occur. As a result, we can only use this idea to infer randomized logspace hardness of the submatrix detection problem for a restricted set of parameters. Those which can be reduced to from $k$-partite planted clique instances with signal strength $k_s$ much much smaller than the threshold at which we believe the problem becomes infeasible. 

Luckily, this still allows us to show randomized logspace hardness for a non-trivial parameter range of submatrix detection problems from the randomized logspace hardness of the $k$-partite planted clique ${\sf{kPC_D}}(\ell \cdot k,k)$ problem when $k = O(\ell^{1-\delta})$ (Conjecture~\ref{conj:logspace-kpc}). See Theorem~\ref{thm:submat-restricted} for a formal statement and discussion about the parameter ranges involved.

\textbf{The remaining challenge:} The restriction on $k_s$ described above means we can not show randomized logspace hardness of the submatrix detection problem for all settings of parameters in which we believe the problem to be hard. To do that, we need to reduce from a $k$-partite planted clique instance ${\sf kPC_D}(\ell \cdot k_s,k_s)$ with $k_s$ only slightly smaller than $\ell$, rather than smaller than $o(\sqrt{\ell})$. 

\subsection{Trying to get more multiple access randomness from the $k$-partite planted clique problem}

A natural attempt to circumvent the restrictions of Lemma~\ref{lem:basic-kpc-randomness-harvesting} is as follows.

Consider a $k$-partite planted clique problem instance ${\sf kPC_D}(\ell \cdot 2k,2k)$ with $k = O(\ell^{1-\delta})$. We do not believe any randomized logspace algorithm can solve this problem (Conjecture~\ref{conj:logspace-kpc}). For some $k_s \leq k$, we could again view the subgraph on the first $\ell \cdot k_s$ vertices of the input as ${\sf kPC_D}(\ell \cdot k_s,k_s)$. Unlike the scheme in Lemma~\ref{lem:basic-kpc-randomness-harvesting} and Figure~\ref{fig:basic-kpc}, we might hope that we can use the edge indicators of the subgraph induced on the last $\ell \cdot k$ vertices as a source of multiple access randomness. The intuition for this is that even though these edge indicators are not independent random bits, efficient algorithms should not be able to distinguish them from independent random bits. This is, after all, what $k$-partite planted clique hardness entails. The advantage now is that we can use as many as ${\ell \cdot k \choose 2}$ (rather than $o(\ell^2 \cdot k)$) multiple access random bits to reduce to other statistical problems from the smaller $k$-partite planted clique instance ${\sf kPC_D}(\ell \cdot k_s,k_s)$. This could potentially allow us to implement the reduction to submatrix detection with larger $k_s$ than in Theorem~\ref{thm:submat-restricted}. See Figure~\ref{fig:advanced-kpc} for a pictorial representation.

\begin{figure}[h]
	\begin{center}
		\begin{tikzpicture}
\draw  (\coord,\coord) to (\coord,0) to (0,\coord) to (\coord,\coord);
\draw [fill] (\coord,\reducrand) to (\coord-\reducrand,\reducrand) to (\coord,0) to (\coord,\reducrand);
\draw [pattern = crosshatch dots] (\reduc,\coord) to (0,\coord) to (\reduc,\coord-\reduc) to (\reduc,\coord) ;
\end{tikzpicture}
\begin{tikzpicture}
		\matrix {
			\node at (0,0) (D) {};
			\node at (0.3*\coord,0.55*\coord) [pattern = crosshatch dots,label=right: Use as problem instance] {};
			\node at (0.3*\coord,0.45*\coord) [fill,label=right: Use as randomness] {};\\
		};
		\end{tikzpicture}
		\caption{Adjacency matrix of the input graph above the diagonal as used in Lemma~\ref{lem:advanced-kpc-randomness-harvesting}.}
		\label{fig:advanced-kpc}
	\end{center}
\end{figure}
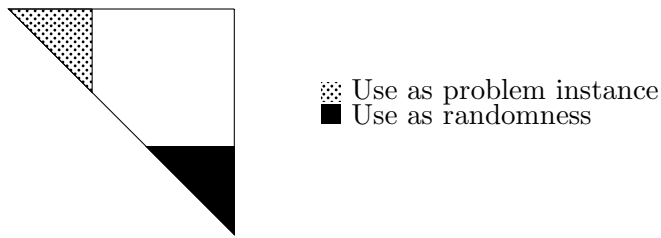

Showing that this idea works involves a slight subtlety. For our harvested randomness to behave as desired, we need to show that no randomized logspace algorithm can distinguish between the following. An Erd\H{o}s-R\'enyi graph ${\sf G}(\ell \cdot k,1/2)$ (which represents what we wish our multiple access randomness was) and a $k$-partite planted clique graph ${\sf kG}(\ell \cdot k,1/2,k)$ (Definition~\ref{def:kpc}) (which represents what we use as our multiple access randomness) \textit{even if} the algorithm has multiple access to another independent $k$-partite planted clique graph ${\sf kG}(\ell \cdot k_s,1/2,k_s)$ (which is the problem instance we reduce from). Unfortunately, the only way we know of doing this is by showing that this auxiliary graph can be simulated by $\Theta((\ell \cdot k_s)^2)$ multiple access random bits and using Lemma~\ref{lem:basic-kpc-randomness-harvesting}. However, we can only invoke Lemma~\ref{lem:basic-kpc-randomness-harvesting} if the amount of auxiliary randomness needed, $\Theta((\ell \cdot k_s)^2)$, is appropriately small. This is, after all, the drawback to Lemma~\ref{lem:basic-kpc-randomness-harvesting} we are trying to circumvent in the first place.

The upshot is that this idea does actually work, but once again with a restriction on $k_s$.
\begin{quote}
	(Lemma~\ref{lem:advanced-kpc-randomness-harvesting}, Informal) If no randomized logspace algorithm can solve the $k$-partite planted clique problem ${\sf{kPC_D}}(\ell \cdot 2k,2k)$, no multiple access randomized logspace algorithm using at most ${\ell \cdot k \choose 2}$ multiple access random bits can solve ${\sf{kPC_D}}(\ell \cdot k_s,k_s)$ for `nice enough' values of $k_s \leq \frac{k}{\sqrt{\ell}}$.
\end{quote}
We need $k=O(\ell^{1-\delta})$ for the conjectured non-existence of randomized logspace algorithms solving ${\sf{kPC_D}}(\ell \cdot 2k,2k)$. This implies a restriction on $k_s$ once again. We can not infer randomized logspace hardness of other statistical problems that we reduce to from a $k$-partite planted clique instance ${\sf{kPC_D}}(\ell \cdot k_s,k_s)$ with $k_s = \Omega(\sqrt{\ell})$. More ideas are needed to make progress on the reduction to submatrix detection for all relevant parameter ranges.

While the scheme of Lemma~\ref{lem:advanced-kpc-randomness-harvesting} does not help us in our goal of reducing to submatrix detection, it is a non-trivial improvement on Lemma~\ref{lem:basic-kpc-randomness-harvesting} in the regime where $k_s$ is small. There could be applications involving hardness of other statistical problems we have not considered where Lemma~\ref{lem:advanced-kpc-randomness-harvesting} is useful even though Lemma~\ref{lem:basic-kpc-randomness-harvesting} is not. Most known reductions between statistical problems are statistical. That is, the reduction maps the input problem to something that is statistically indistinguishable from the output problem. Interestingly, Lemma~\ref{lem:advanced-kpc-randomness-harvesting} is a reduction that itself has a computational aspect.

\subsection{Clique leakage $k$-partite planted clique: hardness of submatrix detection}

Examining the utility of $k$-partite planted clique hardness over planted clique hardness in harvesting multiple access randomness suggests the second theme in this work.
\begin{quote}
\textit{Randomized logspace hardness of more structured variants of the planted clique problem (i.e. where the algorithm solving the problem has more information) allows us to harvest larger amounts of multiple access randomness from the input.}
\end{quote}
\cite{pmlr-v125-brennan20a} pioneered the idea that more structured ``secret leakage''\footnote{\cite{pmlr-v125-brennan20a} introduced the terminology secret leakage because the more structured variants can all be seen as instances of the planted clique problem with additional side information about the location of the planted clique vertices.} variants of the planted clique problem seem to allow polynomial time reducibility to a wider array of statistical problems than from the vanilla planted clique problem. This is why they studied the $k$-partite planted clique problem and its hardness. Our observation is that even the seemingly unrelated task of harvesting larger amounts of multiple access randomness from the input is made easier by assuming hardness of the planted clique problem despite secret leakage.

Our goal, then, is to discover structure or secret leakage that balances two competing requirements. First, the planted clique problem should still be computationally hard despite the secret leakage. Second, the secret leakage should be informative enough to help in randomness harvesting.

We propose the clique leakage $k$-partite planted clique problem ${\sf{clkPC_D}}(\ell \cdot k, k)$ (Definition~\ref{def:clkpc-detection}). As its name suggests, this is the $k$-partite planted clique problem with the additional constraint that the first vertex is always in the planted clique. The algorithm knows that if the input is a planted instance rather than a null instance, the first vertex is in the planted clique. We discuss evidence for the randomized logspace hardness of this problem in Section~\ref{subsec:security-of-cl}. For now, we focus on using this conjectured hardness to show randomized logspace hardness of the submatrix detection problem in full generality.

Consider a clique leakage $k$-partite planted clique problem ${\sf{clkPC_D}}(\ell \cdot 4k, 4k)$. Our belief is that no randomized logspace algorithm can solve this problem when $k = O(\ell^{1-\delta})$ for some constant $\delta > 0$ (Conjecture~\ref{conj:logspace-clkpc}). For some $k_s \leq k$, consider the subgraph on the first $\ell \cdot k_s$ vertices and note that it is an instance of ${\sf{clkPC_D}}(\ell \cdot k_s, k_s)$. Now consider the subgraph induced on the subset of the last $3k$ vertices which \textit{do not} have an edge to the first vertex. We interpret the edge indicator bits of this second subgraph as multiple access randomness. We can use this to simulate any multiple access randomized logspace algorithm on the smaller instance ${\sf{clkPC_D}}(\ell \cdot k_s, k_s)$ that requires at most ${\ell \cdot k \choose 2}$ multiple access random bits. Because the first vertex is guaranteed to be a planted clique vertex if the original large graph is a planted instance, the edge indicator bits we use are always independent random bits. Further, with high probability, the subgraph we use as our source of multiple access randomness will have at least $\ell \cdot k$ vertices, which means we have access to enough multiple access random bits. See Figure~\ref{fig:clkpc} for a pictorial representation of this scheme.

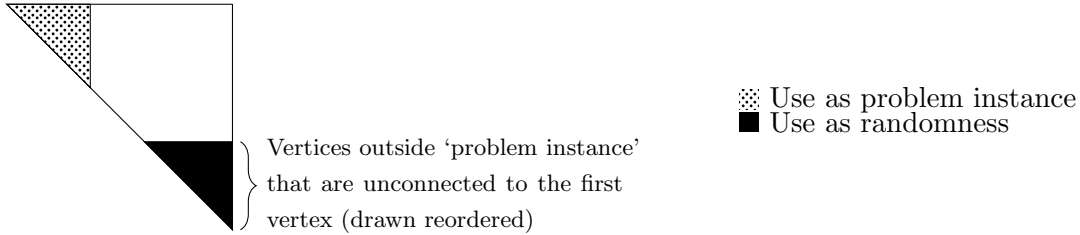
\begin{figure}[h]
	\begin{center}
		\begin{tikzpicture}
		\draw  (\coord,\coord) to (\coord,0) to (0,\coord) to (\coord,\coord);
		\draw [fill] (\coord,\reducrand) to (\coord-\reducrand,\reducrand) to (\coord,0) to (\coord,\reducrand);
		\draw [decorate,decoration={brace,amplitude=2*\coord pt, mirror},xshift=1*\coord pt,yshift=0, text width = 5cm]
		(\coord,0) -- (\coord,\reducrand) node [black,midway,xshift=+2.85cm,yshift=+0cm]
		{\footnotesize Vertices outside `problem instance' that are unconnected to the first vertex (drawn reordered)};
		\draw [pattern = crosshatch dots] (\reduc,\coord) to (0,\coord) to (\reduc,\coord-\reduc) to (\reduc,\coord) ;
		\end{tikzpicture}
		\begin{tikzpicture}
		\matrix {
			\node at (0,0) (D) {};
			\node at (0.3*\coord,0.55*\coord) [pattern = crosshatch dots,label=right: Use as problem instance] {};
			\node at (0.3*\coord,0.45*\coord) [fill,label=right: Use as randomness] {};\\
		};
		\end{tikzpicture}
		\caption{Adjacency matrix of the input graph above the diagonal as used in Lemma~\ref{lem:clkpc-randomness-harvesting}. The vertices outside the `problem instance' have been drawn reordered, so that the vertices unconnected to the first vertex of the graph are all at the end.}
		\label{fig:clkpc}
	\end{center}
\end{figure}

\begin{quote}
	(Lemma~\ref{lem:clkpc-randomness-harvesting}, Informal) If no randomized logspace algorithm can solve the clique leakage $k$-partite planted clique problem ${\sf{clkPC_D}}(\ell \cdot 4k,4k)$, no multiple access randomized logspace algorithm using at most ${\ell \cdot k \choose 2}$ multiple access random bits can solve ${\sf{clkPC_D}}(\ell \cdot k_s,k_s)$ for `nice enough' values of $k_s \leq k$.
\end{quote}
This lets us rule out multiple access randomized logspace algorithms for ${\sf{clkPC_D}}(\ell \cdot k_s,k_s)$, and hence also ${\sf{PC_D}}(\ell \cdot k_s,k_s)$ (see Lemma~\ref{lem:kpc-to-pc}) that use much larger amounts of multiple access randomness than Lemma~\ref{lem:basic-kpc-randomness-harvesting}, without the restriction on $k_s$ needed for Lemma~\ref{lem:advanced-kpc-randomness-harvesting}. This improvement is enough to let us show randomized logspace hardness of the submatrix detection problem for all scalings of parameters in which we believe this problem to be hard. See Theorem~\ref{thm:submat-full} for a formal statement and proof.

\textbf{The remaining challenge:} Lemma~\ref{lem:clkpc-randomness-harvesting} does not let us rule out multiple access randomized algorithms that use arbitrarily large polynomial amounts of multiple access randomness from solving the planted clique problem. Consider the reduction to testing almost $k$-wise independence \cite{alon2007testing} and the reduction to Sparse PCA without the restrictions of Theorem~\ref{thm:spca-restricted}. These only yield multiple access randomized logspace algorithms solving the planted clique problem that use arbitrarily large polynomial amounts of multiple access randomness. To show randomized logspace hardness of these problems, we introduce even more structure in our hardness assumption.

\subsection{Clique leakage hypergraph planted clique: hardness of testing almost $k$-wise independence and Sparse PCA}
\cite{pmlr-v125-brennan20a} studied a ($k$-partite version of a) hypergraph planted clique problem. They showed that its polynomial time hardness implied the polynomial time hardness of all the other secret leakage planted clique problems they used. Further, it allowed them to reduce to a wider array of statistical problems. To provide evidence for the computational hardness of this problem, they showed that low-degree algorithms can not solve it. We show that adding a little more structure to this problem also makes it useful for our randomness harvesting task. Taking inspiration from our previous success with clique leakage, this additional structure too will involve leaking some clique vertex locations to the algorithm.

The hypergraph planted clique problem ${\sf HPC}_{\sf D}^s(n,k)$ (Definition~\ref{def:hpc-detection}) concerns an $s$-uniform hypergraph on $n$ vertices. Here $s \geq 3$ can be any constant positive integer. In the null instance ${\sf HG}^s(n, 1/2)$, every possible $s$-uniform hyperedge is included in this hypergraph independently with probability $1/2$. The alternate instance ${\sf HG}^s(n, 1/2,k)$, is an instance of ${\sf HG}^s(n, 1/2)$ with an additional $s$-uniform clique planted on a uniformly random subset of $k$ vertices.

We propose the clique leakage hypergraph planted clique problem ${\sf clHPC}_{\sf D}^s(n,k)$ (Definition~\ref{def:clhpc-detection}), which is the same problem as above with an additional constraint. In the planted instance of the problem, the first $s-2$ (which is a constant) vertices are guaranteed to be in the planted clique. Our belief is that no randomized logspace algorithm can solve this problem if $k=O(n^{\frac{1}{2}-\delta})$ for some constant $\delta > 0$ (Conjecture~\ref{conj:logspace-clhpc}). We provide evidence for this belief in Section~\ref{subsec:security-of-cl} and focus on using this belief to achieve our desired randomness harvesting goals.

Consider an instance of ${\sf clHPC}_{\sf D}^s(n+s-1,k+s-2)$. It is easy to observe that the subgraph induced by this hypergraph on the $n$ vertices $\{s-1,s,...,n+s-2\}$ is an instance of ${\sf PC_D}(n,k)$. Further, with high probability, the last vertex in the hypergraph is not in the planted clique. We can use the hyperedge indicators for all ${n+s-2 \choose s-1}$ possible hyperedges involving this last vertex as multiple access randomness to solve the aforementioned instance of ${\sf PC_D}(n,k)$. Since $s$ can be any constant we desire, this lets us rule out multiple access randomized logspace algorithms solving the planted clique problem irrespective of the amount of multiple access randomness needed.

\begin{quote}
	(Lemma~\ref{lem:clhpc-randomness-harvesting}, Informal) If no randomized logspace algorithm can solve the clique leakage hypergraph planted clique problem ${\sf clHPC}_{\sf D}^s(n+s-1,k+s-2)$ for any constant integer $s \geq 3$, no multiple access randomized logspace algorithm can solve ${\sf{PC_D}}(n,k)$.
\end{quote}

We know that randomized logspace algorithms for testing almost $k$-wise independence (Lemma~\ref{lem:k-wise-redn}) and Sparse PCA (Lemma~\ref{lem: spca-br13-redn}) imply a multiple access randomized logspace that solves the planted clique problem. Lemma~\ref{lem:clhpc-randomness-harvesting} implies that the aforementioned statistical problems can not be solved by a randomized logspace algorithm. See Theorems~\ref{thm:k-wise} and \ref{thm:spca-full} for formal statements.

We have now shown that the randomized logspace hardness of Sparse PCA \cite{berthet2013complexity}, submatrix detection \cite{ma2015computational}, and testing almost $k$-wise independence \cite{alon2007testing} follows from the randomized logspace hardness (rather than polynomial time hardness) of the planted clique problem or its variants. Table~\ref{table:summary} contains a summary of our results.

\newcolumntype{C}{>{\Centering\arraybackslash}X} 

\begin{table}
	\setlength\extrarowheight{3pt} 
	\begin{tabularx}{\textwidth}{||C|C||}
		\hline
		\textbf{Conjecture} & \textbf{Statistical problem}  \\ \hline
		Logspace planted clique conjecture (Conjecture~\ref{conj:logspace-pc}) & Theorem~\ref{thm:spca-restricted}: Sparse PCA (restricted parameter range) \\
		\hline
		Logspace $k$-partite planted clique conjecture (Conjecture~\ref{conj:logspace-kpc}) & Theorem~\ref{thm:submat-restricted}: Submatrix detection (restricted parameter range) \\
		\hline
		 Clique leakage logspace $k$-partite planted clique conjecture (Conjecture~\ref{conj:logspace-clkpc}) & Theorem~\ref{thm:submat-full}: Submatrix detection \\
		\hline
		Clique leakage logspace hypergraph planted clique conjecture (Conjecture~\ref{conj:logspace-clhpc}) & Theorem~\ref{thm:k-wise}: Testing almost $k$-wise independence and Theorem~\ref{thm:spca-full}: Sparse PCA \\
		\hline
	\end{tabularx}
	\caption{A summary of our results. The conjecture in the first column implies the randomized logspace hardness of the statistical problem shown in the second column.}
	\label{table:summary}
\end{table}

\subsection{Are the $k$-partite and hypegraph planted clique problems computationally hard even with clique leakage?}
\label{subsec:security-of-cl}

Assuming randomized logspace hardness of the clique leakage $k$-partite and clique leakage hypergraph versions of the planted clique problem is useful in showing randomized logspace hardness of other statistical problems. But is this hardness assumption true?

First, \cite{pmlr-v125-brennan20a} provided evidence (by studying low-degree algorithms) that if there is no clique leakage, the $k$-partite and hypergraph versions of the planted clique problem are indeed computationally hard\footnote{Technically, \cite{pmlr-v125-brennan20a} studied the hardness of the $k$-partite version of the hypergraph planted clique problem. However, the hardness of this problem can easily be transferred to the hypergraph planted clique problem with a small amount of multiple access randomness, similar to Lemma~\ref{lem:kpc-to-pc}.}. Clique leakage, which tells the algorithm the location of only a constant number of planted clique vertices, should not make the problem much easier.

A simple idea (also used in \cite{alon1998finding}) can show that a deterministic logspace algorithm solving the \textit{recovery}\footnote{In the recovery task, the input is a (hyper)graph that is guaranteed to have a planted clique. The logspace algorithm must decide, for every vertex, if that vertex is in the planted clique or not. Any reasonable notion of `solving' the recovery task seems to require some form of determinism in the algorithm. Using such an algorithm even to simply verify its correctness (i.e. checking if the vertices it says are in the clique indeed form a clique) requires running it multiple times on the same input and vertex combination. We would want it to output the same answer every time.} problem for clique leakage $k$-partite or clique leakage hypergraph planted clique implies a deterministic logspace algorithm for the detection version of these problems with no clique leakage. Suppose $\omega((\log n)^{\frac{1}{s-1}}) = k = O(n^{\frac{1}{2}-\delta})$ for some constants $s \geq 3$ and $\delta > 0$. If we had a deterministic algorithm to solve clique leakage hypergraph planted recovery (Definition~\ref{defn:recovery-planted-clique-variants}), we could solve hypergraph planted clique detection ${\sf HPC}_{\sf D}^s(n,k)$ (Definition~\ref{def:hpc-detection}) as follows. Iterate over all vertex sets of size $s-2$ (we can do this in logspace because $s$ is a constant), and run the deterministic clique leakage hypergraph recovery algorithm in each of these iterations by pretending this vertex set is actually in the planted clique. Check if the output of this algorithm is indeed a planted clique of size $k$. If this is true for any of the iterations of our algorithm, say that the input graph is a planted clique instance. Otherwise, say it is a null instance.

This algorithm works because of two reasons. First, a null instance (an $s$-uniform Erd\H{o}s-R\'enyi hypergraph ${\sf HG}^s(n, 1/2)$) does not have a clique of size $k$ with high probability. Second, in the planted instance, at least one of the vertex sets of size $s-2$ we run the clique leakage recovery algorithm with will indeed be a part of the planted clique, so the recovery algorithm will output a planted clique of size $k$ in this iteration. A formal proof that this reduction works can be found in Lemma~\ref{lem:clhpc-recovery-hardness}. A similar reduction also holds for the $k$-partite planted clique problem and its clique leakage variant (Lemma~\ref{lem:clkpc-recovery-hardness}).

Combined with the hardness beliefs from \cite{pmlr-v125-brennan20a}, this implies that the recovery problem for $k$-partite and hypergraph planted clique is hard despite clique leakage. For the vanilla planted clique problem ${\sf PC_D}(n,k)$, polynomial time algorithms for detection are known if an only if polynomial time algorithms for recovery are known. \cite{mardia2020space} showed that this tight connection between the complexity of the detection and recovery variants also holds for space complexity (except possibly for a minuscule range of parameters).

It is then reasonable to believe that even for the secret leakage variants of the planted clique problem, space efficient algorithms for detection exist if and only if deterministic space efficient algorithms exist for recovery. It would be an interesting surprise if this were not true for some secret leakage variant of the planted clique problem. If we believe in this tight connection between the recovery and detection problems, we have evidence for the computational hardness of the detection version of the clique leakage $k$-partite and hypergraph planted clique problems.

\subsection{Future outlook and further discussion}

There are two bottlenecks to implementing randomized reductions in logspace. First, we need the reduction to only involve `simple' computations, which do not require large memory. Second, we need multiple access to the random bits the reduction uses. In this work, we have focused on the latter issue. We demonstrated our ideas using three examples, namely Sparse PCA \cite{berthet2013complexity}, submatrix detection \cite{ma2015computational}, and testing $k$-wise independence \cite{alon2007testing}. One reason for choosing these examples was that the first requirement, the reduction being `simple', is satisfied by already existing reductions to these three problems. Hence we could focus on our ideas to tackle the second issue.

\begin{enumerate}
\item That being said, we have not carefully handpicked the only known examples of reductions from the planted clique problem where this simplicity condition holds \cite{applebaum2010public,hazan2011hard,hajek2015computational,gao2017sparse}. Hence our randomness harvesting ideas should be applicable beyond the examples we have provided.

However, the latest generation of reductions from the planted clique problem and its variants \cite{pmlr-v75-brennan18a,brennan2019optimal,pmlr-v125-brennan20a} do not seem to satisfy this simplicity condition. A glance at the reductions implemented in these works reveals some obstacles to obtaining logspace versions of their results\footnote{One not so serious obstacle is that these works use a `real-valued' model of computation rather than a more standard model of computation where inputs are represented using bits. To port these results to logspace, one would either have to use an appropriate `real-valued' model of computation that meaningfully captures space bounds, or use ideas like in \cite{ma2015computational} to relate the real-valued problems to asymptotically equivalent discretized problems.}. Some of the primitives require taking the `square-root' of matrices (after checking their positive semidefiniteness), or sampling a Haar random matrix from the orthogonal group. Standard implementations of either of these tasks use linear algebraic techniques that require large amounts of memory. It is an intriguing open problem to see if the hardness results in these works can be made robust to changing the notion of computational efficiency to space efficiency. This might be either by implementing known reductions more efficiently\footnote{To demonstrate that there is indeed hope of doing this, we sketch a heuristic idea to bypass the difficulty in taking square roots of matrices. Carefully inspecting the reduction CLIQUE-TO-WISHART in \cite{brennan2019optimal}, we see that the matrix we want a square root of is a small perturbation of the identity matrix. It is plausible, though not certain, that the reduction goes through even if we replace the square root function by a truncated taylor series of the square root function about the identity matrix (for example, $\sqrt{1+x} \approx 1+\frac{x}{2}$). It would be easy to implement such a function space efficiently.}, or designing alternate reductions that are space efficient.

\item Like \cite{pmlr-v125-brennan20a}, we too see a tradeoff between the conservativeness of our hardness assumption and the utility of that assumption. In our context, we could harvest more multiple access randomness if we assumed computational hardness of more structured variants of the planted clique problem. This suggests two directions of study. The first is to provide more evidence that the secret leakage variants whose hardness we have assumed are indeed computationally hard, particularly the clique leakage variants. The second is to try and devise cleverer techniques to harvest more multiple access randomness from less structured versions of the planted clique problem. That is, versions in which there is less secret leakage. Lemma~\ref{lem:advanced-kpc-randomness-harvesting} is an attempt of this sort.
\end{enumerate}

Harvesting multiple access randomness from (all variants of) the planted clique problem turned out to be particularly easy. In all our schemes, we simply used bits from the input as multiple access randomness without any post-processing. This is unlike works that use randomness harvesting for derandomization, where the the input is often passed through sophisticated randomness processing tools \cite{goldreich2002derandomization,kinne2012pseudorandom,hoza2017typically,Chen2021HardnessVR}. If we believe that the ability to transfer computational hardness for statistical problems should be robust to changing our notion of reducibility from polynomial time to logarithmic space, this is further evidence for the utility of the planted clique problem and its variants. Not all problems would be so amenable to harvesting randomness from the input.

\section{Related Work}
\label{sec:related-work}

\textbf{The planted clique problem is believed to be hard, and this hardness is useful:} It is widely believed that polynomial time algorithms can only detect or recover the planted clique for clique sizes above $k=\Omega(\sqrt{n})$. One piece of evidence for this belief is the long line of algorithmic progress using a variety of techniques that has been unable to break this barrier \cite{kuvcera1995expected,alon1998finding,feige2000finding,feige2010finding,ames2011nuclear,dekel2014finding,chen2014statistical,deshpande2015,hajek2015computational,mardia2020finding}. The other piece of evidence comes from studying restricted but powerful classes of algorithms. \cite{jerrum1992large} showed that a natural Markov chain based technique requires more than polynomial time below this threshold. Similar hardness results for the planted clique problem or its variants have been shown for statistical query algorithms \cite{feldman2017statistical}, circuit classes \cite{rossman2008constant,rossman2010monotone}, the Lov{\'a}sz--Schrijver hierarchy \cite{feige2003probable}, and the sum-of-squares hierarchy \cite{meka2015sum,deshpande2015improved,hopkins2018integrality,barak2019nearly}. 

Further evidence comes from the low-degree-likelihood method \cite{hopkins2017bayesian,hopkins2017power,hopkins2018statistical,kunisky2019notes} and through concepts from statistical physics \cite{gamarnik2019landscape}. This conjectured polynomial time hardness of the planted clique problem has also been widely used to show the polynomial time hardness of other problems, and we refer the reader to \cite{pmlr-v125-brennan20a} for a comprehensive list of examples. In the spirit of studying the utility of different notions of `efficient computation', in this work we show interesting reductions from \textit{weaker} computational assumptions (space rather than time). In the same spirit but in the opposite direction, \cite{manurangsi2021strongish} demonstrated the utility of \textit{stronger} computational assumptions for the planted clique problem.

\textbf{Sparse PCA, submatrix detection, and testing almost $k$-wise independence:}
The Sparse PCA problem has been well-studied from a computational lower bound perspective \cite{berthet2013complexity,berthet2013optimal,wang2016statistical,brennan2019optimal}. See \cite{ding2019subexponential} for a discussion about the extensive algorithmic progress on the Sparse PCA problem. Submatrix detection \cite{butucea2013detection,montanari2015limitation,kolar2011minimax,balakrishnan2011statistical,chen2016statistical,cai2017computational,ma2015computational,brennan2019universality} and testing almost $k$-wise independence \cite{alon2007testing,o2018closeness} have also been widely studied.

\textbf{The hypergraph planted clique problem:}
Computational problems involving the Erd\H{o}s-R\'enyi or planted clique hypergraphs have also been studied before \cite{pmlr-v125-brennan20a}, though less comprehensively than the corresponding graph problems. \cite{boix2019average} studied the average-case complexity of counting cliques in Erd\H{o}s-R\'enyi hypergraphs. \cite{zhang2018tensor,luo2020tensor} used the conjectured hardness of the hypergraph planted clique problem to show hardness of other statistical problems, and the latter also provided evidence for this conjecture by studying Markov chain based algorithms like in \cite{jerrum1992large}.

\section{Technical Preliminaries}
\label{sec:defns+tech_prelim}

Unless stated otherwise, all logarithms are taken base $2$. $[n] := \{1,2,...,n\}$.\\

\begin{definition}[Randomized logspace algorithm]\label{defn:bdd-space-alg}\ \\
A randomized logspace algorithm $\A$ is a probabilistic random-access Turing Machine. Given multiple access to a binary input $x$ of length $n$ and read-once access to a uniformly random bit string $r_{ro}$ of length at most $\poly(n)$, it can compute the output bit $\A(x,r_{ro})$ using at most $O(\log n)$ bits of multiple access working space and $\poly(n)$ time.\\
\end{definition}

\begin{definition}[Multiple access randomized logspace algorithm]\label{defn:multacc-bdd-space-alg}\ \\
A multiple access randomized logspace algorithm $\A$ is a probabilistic random-access Turing Machine. Given multiple access to a binary input $x$ of length $n$, read-once access to a uniformly random bit string $r_{ro}$, and multiple access to a uniformly random bit string $r_{ma}$ (each of length at most $\poly(n)$) it can compute the output bit $\A(x,r_{ro},r_{ma})$ using at most $O(\log n)$ bits of multiple access working space and $\poly(n)$ time.\\
\end{definition}

\begin{definition}[Hypothesis testing problem]\label{defn:hyp-test}\ \\
A hypothesis testing problem $\mathcal{P}$ consists of non-decreasing positive integer sequences $\ell = \omega(1)$ and ${\sf size_{\mathcal{P}}}(\ell)$ along with two sequences of hypotheses ${\sf H}_0(\ell)$ and ${\sf H}_1(\ell)$. For $i \in \{0,1\}$, ${\sf H}_i(\ell)$ is a set of distributions over binary strings of length ${\sf size_{\mathcal{P}}}(\ell)$. For a problem indexed by $\ell$, the input $x \in \{0,1\}^{{\sf size_{\mathcal{P}}}(\ell)}$ is distributed with probability $1/2$ according to some distribution $P_0 \in {\sf H}_0(\ell)$ or with probability $1/2$ according to some distribution $P_1 \in {\sf H}_1(\ell)$.\\
\end{definition}

\begin{definition}[Solving a hypothesis testing problem]\label{defn:solve-hyp-test}\ \\
A randomized algorithm $\A$ with access to a uniformly random bit string $r$ solves a hypothesis testing problem $\mathcal{P}$ if
\[\inf\limits_{P_0 \in {\sf H}_0(\ell)}\underset{\hspace{-1em}\mathrlap{x \sim P_0, r}}{\Prob}\left(\A(x,r) = 0 \right) + \inf\limits_{P_1 \in {\sf H}_1(\ell)}\underset{\hspace{-1em}\mathrlap{x \sim P_1, r}}{\Prob}\left(\A(x,r) = 1 \right) = 1+ \Omega(1).\]
\end{definition}

\begin{definition}[Erd\H{o}s-R\'enyi graph distribution: $\G(n, 1/2)$]\label{def:er}\ \\
Let $G = ([n],E)$ be a graph with vertex set of size $n$. The edge set $E$ is created by including each possible edge independently with probability $\frac{1}{2}$. The distribution on graphs (represented by their $n^2$-bit adjacency matrix) thus formed is denoted $\G(n, 1/2)$.\\
\end{definition}

\begin{definition}[Planted Clique graph distribution: $\G(n, 1/2,k)$]\label{def:pc}\ \\
Let $G = ([n],E)$ be a graph with vertex set of size $n$. Let $K \subset [n]$ be a set of size $k$ chosen uniformly at random from all ${n \choose k}$ subsets of size $k$. For all distinct pairs of vertices $u,v \in K$, we add the edge $(u,v)$ to $E$. 
	For all remaining distinct pairs of vertices $u,v$, 
	we add the edge $(u,v)$ to $E$ independently with probability $\frac{1}{2}$. The distribution on graphs (represented by their $n^2$-bit adjacency matrix) thus formed is denoted $\G(n, 1/2,k)$.\\
\end{definition}

\begin{definition}[Planted Clique Detection Problem: ${\sf{PC_D}}(n,k)$]\label{def:pc-detection}\ \\
Given non-decreasing positive integer sequences $n=\omega(1)$ and $k \leq n$, this is the hypothesis testing problem induced by
\begin{align*}
{\sf H}_0(n) = \left\{\G(n,1/2)\right\} \hspace{0.1cm} \text{ and } \hspace{0.1cm}
{\sf H}_1(n) =
\left\{\G(n, 1/2,k)\right\}
\end{align*}
where the input to problem is the triple given by the number of vertices $n$, the size of the planted clique $k$, and the adjacency matrix of the random graph.\\
\end{definition}


\section{Harvesting Multiple Access Randomness from the Input}
\label{sec:rand-harvesting}

\subsection{Planted Clique}
\label{subsec:rand-from-pc}

\begin{conjecture}[Logspace Planted Clique Conjecture: {\sf PC-Conj-Space}]\label{conj:logspace-pc}\ \\
Let $\A$ be any randomized logspace algorithm whose read-once uniformly random bit string is denoted $r_{ro}$. Let $n = \omega(1)$ and $k$ be any sequences of non-decreasing positive integers such that $k = O(n^{\frac{1}{2}-\delta})$ for some constant $\delta > 0$. Then
	\[
	\underset{\hspace{-1em}\mathrlap{A_G \sim \G(n,1/2),~r_{ro}}}{\Prob}\left( \A(\{n,k,A_G\},r_{ro}) = 0 \right) +
	\underset{\hspace{-1em}\mathrlap{A_G \sim \G(n,1/2,k),~r_{ro}}}{\Prob}\left(\A(\{n,k,A_G\},r_{ro}) = 1\right) = 1+o(1)   .\]\\
\end{conjecture}

\begin{lemma}[Small amounts of multiple access randomness from the Planted Clique input]\label{lem:basic-randomness-harvesting}\ \\
Let $n=\omega(1)$, $k$, and $m= o(n/k)$ be any non-decreasing positive integer sequences such that $k=O(n^{\frac{1}{2}-\delta})$ for some positive constant $\delta > 0$. Given $n$ and $k$, assume we can deterministically compute and store $m$ using an additional $O(\log n)$ bits of space. 

If {\sf PC-Conj-Space} (Conjecture~\ref{conj:logspace-pc}) is true, no multiple access randomized logspace algorithm using at most ${n \choose 2} - {n-m \choose 2} = o(n^2 / k)$ multiple access random bits\footnote{This upper bound is required only on the multiple access randomness used. We do not assume any upper bounds on the read-once randomness other than those implied by the space constraints on the algorithm.} can solve ${\sf{PC_D}}(n-m,k)$.
\end{lemma}
\begin{proof}
Let $\mathcal{S}$ be any multiple access randomized logspace algorithm that uses at most ${n \choose 2} - {n-m \choose 2}$ multiple access random bits when run on an instance of ${\sf PC_D}(n-m,k)$. We will show that $\mathcal{S}$ fails to solve ${\sf{PC_D}}(n-m,k)$.

The key idea is to use randomness from the input itself, since we have multiple access to the input. Let $\mathcal{S'}$ be an algorithm which takes an input $\{n,k,A'_G\}$ corresponding to ${\sf{PC_D}}(n,k)$. Let $A_G$ denote the adjacency matrix of the subgraph induced by $A'_G$ (the adjacency matrix input to $\mathcal{S'}$) on its first $n-m$ vertices. Let $r_G$ be a binary string containing the ${n \choose 2} - {n-m \choose 2}$ edge indicator bits for possible edges with at least one endpoint in the last $m$ vertices. Let $\mathcal{S'}$ run $\mathcal{S}$ on $A_G$ using $r_G$ as its source of multiple access random bits and output the answer given by $\mathcal{S}$.

\begin{itemize}[leftmargin=0pt]
\item \textbf{$\mathcal{S'}$ is a randomized logspace algorithm:}
We first make a technical remark. The number of the bits in the input triple $\{$number of vertices, number of planted vertices, adjacency matrix$\}$ for both problems ${\sf{PC_D}}(n-m,k)$ and ${\sf{PC_D}}(n,k)$ are within polynomial factors of each other, since $m = o(n/k) = o(n)$. Thus a logspace algorithm for one is also a logspace algorithm for the other. Hence we can ignore the distinction about which of the two problems we are running the algorithm on when determining its space complexity.

By assumption, $\mathcal{S'}$ can deterministically compute and store $m$ (and hence $n-m$) using $O(\log n)$ bits of working space. This lets $\mathcal{S'}$ deterministically simulate multiple access to $n-m$, $k$, $A_G$, and $r_G$ from multiple access to $n$, $k$, and $A'_G$ using $O(\log n)$ bits of working space. Finally, $\mathcal{S}$ (and hence $\mathcal{S'}$) can be implemented as a randomized logspace algorithm given multiple access to $n-m$, $k$, $A_G$, and $r_G$ because $r_G$ has ${n \choose 2} - {n-m \choose 2}$ bits.

Because $k = O(n^{\frac{1}{2}-\delta})$, {\sf PC-Conj-Space} (Conjecture~\ref{conj:logspace-pc}) combines with the fact above to imply that $\mathcal{S'}$ fails to solve ${\sf PC_D}(n,k)$. Formally, \[\underset{\hspace{-1em}\mathrlap{A'_G \sim \G(n,1/2),r_{ro}}}{\Prob}\left( \mathcal{S'}(\{n,k,A'_G\},r_{ro}) = 0 \right) +
\underset{\hspace{-1em}\mathrlap{A'_G \sim \G(n,1/2,k),r_{ro}}}{\Prob}\left(\mathcal{S'}(\{n,k,A'_G\},r_{ro}) = 1\right)=1+o(1).\]

\item \textbf{$\mathcal{S}$ fails to solve ${\sf{PC_D}}(n-m,k)$:} Let ${\sf LastVerticesNotInClique}$ denote the event that under the random choice of clique vertices in ${\sf{PC_D}}(n,k)$, all $k$ planted clique vertices are in the first $n-m$ vertices. Conditioned on ${\sf LastVerticesNotInClique}$, it is clear that $\mathcal{S}$ (when invoked by $\mathcal{S'}$) is being run on an instance of ${\sf{PC_D}}(n-m,k)$. This is because 
\begin{enumerate}
\item Irrespective of whether the input graph is $\G(n,1/2)$ or $\G(n,1/2,k)$, $r_G$ consists of iid uniform random bits that are independent of $A_G$.
\item If $A'_G \sim \G(n,1/2)$, then $A_G \sim \G(n-m,1/2)$. On the other hand, if $A'_G \sim \G(n,1/2,k)$, then $A_G \sim \G(n-m,1/2,k)$.
\end{enumerate}
Conditioned on {\sf LastVerticesNotInClique}, the probability that $\mathcal{S'}$ outputs the correct answer is equal to the probability that $\mathcal{S}$ outputs the correct answer when run on a instance of ${\sf PC_D}(n-m,k)$. Because $m = o(n/k)$, a union bound gives $\Pr({\sf LastVerticesNotInClique}^c) \leq m \cdot \frac{k}{n} = o(1)$. Hence we must have \[\underset{\hspace{-1em}\mathrlap{A_G \sim \G(n-m,1/2),r_{ro},r_{ma}}}{\Prob}\left( \mathcal{S}(\{n-m,k,A_G\},r_{ro},r_{ma}) = 0 \right) + \underset{\hspace{-1em}\mathrlap{A_G \sim \G(n-m,1/2,k),r_{ro},r_{ma}}}{\Prob}\left(\mathcal{S}(\{n-m,k,A_G\},r_{ro},r_{ma}) = 1\right)=1+o(1).\qedhere\]
\end{itemize}

\end{proof}

\subsection{$k$-Partite Planted Clique}
\label{subsec:rand-from-kpc}

\begin{definition}[$k$-Partite Planted Clique graph distribution: ${\sf kG}(\ell \cdot k, 1/2,k)$]\label{def:kpc}\ \\
This is the planted clique graph distribution ${\sf G}(\ell \cdot k, 1/2,k)$ conditioned on the event that for all $i \in \{1,2,...,k\}$, there is exactly one planted clique vertex in $[i \cdot \ell] \setminus [(i-1) \cdot \ell]$.\\
\end{definition}

\begin{definition}[$k$-Partite Planted Clique Detection Problem: ${\sf{kPC_D}}(\ell \cdot k, k)$]\label{def:kpc-detection}\ \\
Given non-decreasing positive integer sequences $\ell = \omega(1)$ and $k$, this is the hypothesis testing problem induced by \[{\sf H}_0(\ell) = \left\{{\sf G}( \ell \cdot k,1/2)\right\} \hspace{0.1cm} \text{ and } \hspace{0.1cm}
{\sf H}_1(\ell) = \left\{{\sf kG}(\ell \cdot k, 1/2,k)\right\}\]
where the input to problem is the triple given by the number of vertices $\ell \cdot k$, the size of the planted clique $k$, and the adjacency matrix of the random graph.\\
\end{definition}

\begin{conjecture}[Logspace $k$-Partite Planted Clique Conjecture: {\sf kPC-Conj-Space}]\label{conj:logspace-kpc}\ \\
Let $\A$ be any randomized logspace algorithm whose read-once uniformly random bit string is denoted $r_{ro}$. Let $\ell = \omega(1)$ and $k$ be any sequence of non-decreasing positive integers such that $k = O(\ell^{1-\delta})$ for some constant $\delta > 0$. Then
	\[
	\underset{\hspace{-1em}\mathrlap{A_G \sim {\sf G}(\ell \cdot k,1/2),~r_{ro}}}{\Prob}\left( \A(\{\ell \cdot k,k,A_G\},r_{ro}) = 0 \right) +
	\underset{\hspace{-1em}\mathrlap{A_G \sim {\sf kG}(\ell \cdot k,1/2,k),~r_{ro}}}{\Prob}\left(\A(\{\ell \cdot k,k,A_G\},r_{ro}) = 1\right) = 1+o(1)   .\]\\
\end{conjecture}

\begin{lemma}[Multiple access randomness from the $k$-Partite Planted Clique input]\label{lem:basic-kpc-randomness-harvesting}\ \\
Let $\ell = \omega(1)$, $k$, $k_s \leq k-1$, and $m = o(\ell)$ be any non-decreasing positive integer sequences such that $k = O(\ell^{1-\delta})$ for some positive constant $\delta>0$. Given $\ell \cdot k$ and $k$, assume we can deterministically compute and store $k_s$ and $m$ using an additional $O(\log \ell)$ bits of space. 

If {\sf kPC-Conj-Space} (Conjecture~\ref{conj:logspace-kpc}) is true, no multiple access randomized logspace algorithm using at most ${\ell \cdot k \choose 2} - {\ell \cdot k-m \choose 2} = o(\ell^2 \cdot k)$ multiple access random bits can solve ${\sf{kPC_D}}(\ell \cdot k_s,k_s)$.
\end{lemma}
\begin{proof}
The proof is analogous to that of Lemma~\ref{lem:basic-randomness-harvesting}, so we note the slight change required and omit the details. The only difference from Lemma~\ref{lem:basic-randomness-harvesting} is that the reduction uses the subgraph induced on the first $\ell \cdot k_s$ vertices (instead of the first $n-m$ vertices) as the instance of ${\sf{kPC_D}}(\ell \cdot k_s,k_s)$ to solve. The auxiliary multiple access random bits harvested from the input are still the edge indicators for possible edges with at least one endpoint in the last $m$ vertices. The key reason this works is that the $k$-Partite planted clique problem has a nice `self-reducibility' structure so that the subgraph induced on the first $\ell \cdot k_s$ vertices is an instance of ${\sf{kPC_D}}(\ell \cdot k_s,k_s)$. \qedhere
\end{proof}

\begin{lemma}[More multiple access randomness from the $k$-Partite Planted Clique input]\label{lem:advanced-kpc-randomness-harvesting}\ \\
Let $\ell = \omega(1)$, $k$ and $k_s$ be any non-decreasing positive integer sequences such that $k  = O(\ell^{1-\delta})$ for some constant $\delta > 0$ and $k_s \leq \frac{k}{\sqrt{\ell}} $. Given $\ell \cdot k$ and $k$, assume we can deterministically compute and store $k_s$ using an additional $O(\log \ell)$ bits of space. 
	
If the {\sf kPC-Conj-Space} (Conjecture~\ref{conj:logspace-kpc}) is true, no multiple access randomized logspace algorithm using at most ${\ell \cdot k \choose 2}$ multiple access random bits can solve ${\sf{kPC_D}}(\ell \cdot k_s,k_s)$.
\end{lemma}
\begin{proof}
Let $\mathcal{S}$ be any multiple access randomized logspace algorithm that uses at most ${\ell \cdot k \choose 2}$ multiple access random bits when run on an instance of ${\sf{kPC_D}}(\ell \cdot k_s,k_s)$. We will show that $\mathcal{S}$ must fail to solve ${\sf{kPC_D}}(\ell \cdot k_s,k_s)$. 

Let $\mathcal{S'}$ be an algorithm which takes as input $\{\ell \cdot 2k,2k,A_G\}$ corresponding to ${\sf{kPC_D}}(\ell \cdot 2k,2k)$. Let $A_s$ (respectively $A_r$) denote the adjacency matrix of the subgraph induced by $A_G$ (the adjacency matrix input to $\mathcal{S'}$) on its first $\ell \cdot k_s$ (respectively last $\ell \cdot k$) vertices. Let $r_G$ be a binary string containing the ${\ell \cdot k \choose 2}$ edge indicator bits for possible edges in the graph corresponding to $A_r$. Let $\mathcal{S'}$ run $\mathcal{S}$ on $A_s$ using $r_G$ as its source of at most ${\ell \cdot k \choose 2}$ multiple access random bits and output the answer given by $\mathcal{S}$.

\begin{itemize}[leftmargin=0pt]
\item \textbf{$\mathcal{S'}$ is a randomized logspace algorithm:} The space complexity analysis of $\mathcal{S'}$ is like in the proof of Lemma~\ref{lem:basic-randomness-harvesting}. We omit the details.
\end{itemize}

Combined with the {\sf kPC-Conj-Space} (Conjecture~\ref{conj:logspace-kpc}) and the fact that $2k = O(\ell^{1-\delta})$, this means $\mathcal{S'}$ fails to solve ${\sf{kPC_D}}(\ell \cdot 2k,2k)$. Formally, \[\underset{\hspace{-1em}\mathrlap{A_G \sim \G(\ell \cdot 2k,1/2),r_{ro}}}{\Prob}\left( \mathcal{S'}(\{\ell \cdot 2k,2k,A_G\},r_{ro}) = 0 \right) +
\underset{\hspace{-1em}\mathrlap{A_G \sim {\sf kG}(\ell \cdot 2k,1/2,2k),r_{ro}}}{\Prob}\left(\mathcal{S'}(\{\ell \cdot 2k,2k,A_G\},r_{ro}) = 1\right) = 1+o(1).\]

\begin{itemize}[leftmargin=0pt]
\item \textbf{$\mathcal{S}$ fails to solve ${\sf{kPC_D}}(\ell \cdot k_s,k_s)$:}
Letting $r_{ma}$ denote a set of ${\ell \cdot k \choose 2}$ independent uniformly random multiple access bits, we want to show that \[\underset{\hspace{-1em}\mathrlap{A_s \sim {\sf G}(\ell \cdot k_s,1/2),r_{ro},r_{ma}}}{\Prob}\left( \mathcal{S}(\{\ell \cdot k_s,k_s,A_s\},r_{ro},r_{ma}) = 0 \right) +
\underset{\hspace{-1em}\mathrlap{A_s \sim {\sf kG}(\ell \cdot k_s,1/2,k_s),r_{ro},r_{ma}}}{\Prob}\left( \mathcal{S}(\{\ell \cdot k_s,k_s,A_s\},r_{ro},r_{ma}) = 1 \right)=1+o(1).\]
\begin{enumerate}
\item If the input graph $A_G$ is drawn from ${\sf G}(\ell \cdot 2k,1/2)$, then $A_s \sim {\sf G}(\ell \cdot k_s,1/2)$ and $r_G$ is a set of independent uniformly random bits. Hence we can relate the performance of $\mathcal{S'}$ to the performance of $\mathcal{S}$ to obtain \[\underset{\hspace{-1em}\mathrlap{A_s \sim {\sf G}(\ell \cdot k_s,1/2),r_{ro},r_{ma}}}{\Prob}\left( \mathcal{S}(\{\ell \cdot k_s,k_s,A_s\},r_{ro},r_{ma}) = 0 \right)=\underset{\hspace{-1em}\mathrlap{A_G \sim \G(\ell \cdot 2k,1/2),r_{ro}}}{\Prob}\left( \mathcal{S'}(\{\ell \cdot 2k,2k,A_G\},r_{ro}) = 0 \right).\]
\item If the input graph $A_G$ is drawn from ${\sf kG}(\ell \cdot 2k,1/2,2k)$, then we do have $A_s \sim {\sf kG}(\ell \cdot k_s,1/2,k_s)$ as we might hope. However, $r_G$ is not a collection of independent uniformly random bits. This is because $A_r \sim {\sf kG}(\ell \cdot k,1/2,k)$. Hence $\mathcal{S'}$ is actually using a source of multiple access random bits that is not truly independent and uniform.
\end{enumerate}

Luckily, efficient algorithms should not be able to distinguish between ${\sf kG}(\ell \cdot k,1/2,k)$ and a source of truly random bits. This allows us to relate the performance of $\mathcal{S'}$ when using $A_r$ as its source of multiple access randomness to its performance when using a truly random source of multiple access bits. If 
\begin{align}
\label{eqn:boost-proof}
\lvert \underbrace{\underset{\hspace{-1em}\mathrlap{A_s \sim {\sf kG}(\ell \cdot k_s,1/2,k_s),r_{ro},r_{ma}}}{\Prob}\left( \mathcal{S}(\{\ell \cdot k_s,k_s,A_s\},r_{ro},r_{ma}) = 1 \right)}_{:=1-\mathcal{S}_0} - \underset{\hspace{-1em}\mathrlap{A_G \sim {\sf kG}(\ell \cdot 2k,1/2,2k),r_{ro}}}{\Prob}\left(\mathcal{S'}(\{\ell \cdot 2k,2k,A_G\},r_{ro}) = 1\right) \rvert  = o(1),
\end{align} 
then clearly we have the desired conclusion that $\mathcal{S}$ fails to solve ${\sf{kPC_D}}(\ell \cdot k_s,k_s)$. We now show that (\ref{eqn:boost-proof}) must hold. Let $r_G$ denote the dependent collection of ${\ell \cdot k \choose 2}$ multiple access random bits obtained from $A_r \sim {\sf kG}(\ell \cdot k,1/2,k)$. By construction, the output of $\mathcal{S'}$ is identical to the output of $\mathcal{S}$ when run with $r_G$ as its source of multiple access randomness, which means \[\underbrace{\underset{\hspace{-1em}\mathrlap{A_s \sim {\sf kG}(\ell \cdot k_s,1/2,k_s),r_{ro},r_{G}}}{\Prob}\left( \mathcal{S}(\{\ell \cdot k_s,k_s,A_s\},r_{ro},r_{G}) = 1 \right)}_{:=\mathcal{S}_1} = \underset{\hspace{-1em}\mathrlap{A_G \sim {\sf kG}(\ell \cdot 2k,1/2,2k),r_{ro}}}{\Prob}\left(\mathcal{S'}(\{\ell \cdot 2k,2k,A_G\},r_{ro}) = 1\right).\] Hence, showing (\ref{eqn:boost-proof}) is equivalent to showing $\mathcal{S}_0 + \mathcal{S}_1 = 1+o(1)$. We will reason about $\mathcal{S}_0$ and $\mathcal{S}_1$ using a hypothetical algorithm $\mathcal{H}$ whose input $\{\ell \cdot k, k,A_{h_G}\}$ is an instance of the problem ${\sf kPC_D}(\ell \cdot k,k)$. $\mathcal{H}$ outputs the answer obtained by running $\mathcal{S}$ using its input $A_{h_G}$ as the multiple access randomness $\mathcal{S}$ needs and $\{\ell \cdot k_s,k_s,A_{h_r}\}$ as the `input' to $\mathcal{S}$. Here $A_{h_r}$ is an adjacency matrix on $\ell \cdot k_s$ vertices we will define later.

We show that $\mathcal{H}$ can be implemented as a multiple access randomized logspace algorithm that uses $r := {\ell \cdot k_s \choose 2} + k_s \cdot (10 \ell (\lceil \log \ell \rceil)^2)$ multiple access random bits. By assumption, $\mathcal{H}$ can compute $\ell \cdot k_s,k_s$ from its input using $O(\log \ell)$ bits of space. It can also clearly provide $\mathcal{S}$ with multiple access to its input adjacency matrix $A_{h_G}$ for use as randomness. Since $\mathcal{S}$ is a multiple access randomized logspace machine and the inputs to all problems we consider are of size $\poly(\ell)$, we only need to show that  multiple access to the (yet to be defined) matrix $A_{h_r}$ takes $O(\log \ell)$ bits of space.

Of the $r$ multiple access random bits, let the first ${\ell \cdot k_s \choose 2}$ be the edge indicators of an Erd\H{o}s-R\'enyi ${\sf G}(\ell \cdot k_s,1/2)$ graph. Partition the rest into $k_s$ sets of $(10 \ell (\lceil \log \ell \rceil)^2)$ independent uniform random bits. For all $j \in [k_s]$, use the $j^{th}$ such set to sample an efficiently computable random function (which is approximately a uniformly random permutation) $\pi_j : [\ell] \rightarrow [\ell]$ using Lemma~\ref{lem:permutation}. Because these functions are all sampled independently, the `sub-additivity' or `tensorization' property of total variation distance \cite[Fact 3.1]{brennan2019optimal} implies that the total variation distance between the set of functions $\pi_j$ obtained from Lemma~\ref{lem:permutation} and a set of independent uniformly random permutations is $O(k_s/\ell^4) = o(1)$.

We now define $A_{h_r}$ so that it is `approximately' distributed as ${\sf kPC_D}(\ell \cdot k_s,1/2,k_s)$. Let it be the adjacency matrix of the $\ell \cdot k_s$ vertex graph whose edge set is as follows. Given two vertices $(j_1-1) \cdot \ell + i_1, (j_2-1) \cdot \ell + i_2$ with $j_1,j_2 \in [k_s]$ and $i_1,i_2 \in [\ell]$, if $\pi_{j_1}(i_1) = \pi_{j_2}(i_2) = 1$, there is an edge between the two vertices. Otherwise, there is an edge between the two vertices if and only if there is an edge between the corresponding two vertices in the Erd\H{o}s-R\'enyi ${\sf G}(\ell \cdot k_s,1/2)$ graph sampled using the the first ${\ell \cdot k_s \choose 2}$ multiple access random bits available to $\mathcal{H}$.

It is straightforward to observe that because the functions $\pi_j$ can all be implemented using an additional $O(\log \ell)$ bits of space (Lemma~\ref{lem:permutation}), the algorithm $\mathcal{H}$ can provide the algorithm $\mathcal{S}$ with multiple access to the adjacency matrix $A_{h_r}$ using $O(\log \ell)$ bits of working space. This completes the proof that $\mathcal{H}$ can be implemented as a multiple access randomized logspace algorithm that uses $r$  multiple access random bits. Because \[r \leq 2{\ell \cdot k_s \choose 2} \leq (\ell\cdot k_s)^2 \leq  \frac{\ell \cdot k \cdot (2k)}{2} \leq  {\ell \cdot (k+1) \choose 2} - {\ell \cdot (k+1)-2k \choose 2},\] Lemma~\ref{lem:basic-kpc-randomness-harvesting} implies that our hypothetical algorithm $\mathcal{H}$ must fail to solve ${\sf kPC_D}(\ell \cdot k,k)$. Formally, this means
\[
\underbrace{\underset{\hspace{-1em}\mathrlap{A_G \sim {\sf G}(\ell \cdot k,1/2),~r_{ro},r_{ma}}}{\Prob}\left( \mathcal{H}(\{\ell \cdot k,k,A_G\},r_{ro},r_{ma}) = 0 \right)}_{:=\mathcal{H}_0} +
\underbrace{\underset{\hspace{-1em}\mathrlap{A_G \sim {\sf kG}(\ell \cdot k,1/2,k),~r_{ro},r_{ma}}}{\Prob}\left(\mathcal{H}(\{\ell \cdot k,k,A_G\},r_{ro},r_{ma}) = 1\right)}_{:=\mathcal{H}_1} = 1+o(1)   .\]

Recall that that the total variation distance between the set of functions $\pi_j$ used by $\mathcal{H}$ and a set of independent uniformly random permutations is $o(1)$. Further, if $\mathcal{H}$ used independent uniformly random permutations instead of the $\pi_j$'s from Lemma~\ref{lem:permutation}, $A_{h_r}$ would be distributed exactly as ${\sf kPC_D}(\ell \cdot k_s,1/2,k_s)$. Because of how $\mathcal{H}$ was constructed, the data processing inequality for total variation distance \cite[Fact 3.1]{brennan2019optimal} thus implies $\lvert \mathcal{S}_i-\mathcal{H}_i \rvert = o(1)$ for $i \in \{0,1\}$. This completes our proof.\qedhere
\end{itemize}
\end{proof}

\subsection{Clique Leakage $k$-Partite Planted Clique}
\label{subsec:rand-from-clkpc}

\begin{definition}[Clique Leakage $k$-Partite PC graph distribution: ${\sf clkG}(\ell \cdot k, 1/2,k)$]\label{def:clkpc}\ \\
This is the $k$-partite planted clique graph distribution ${\sf kG}(\ell \cdot k, 1/2,k)$ conditioned on the event that the first vertex in the graph (i.e the vertex named $1$) is in the set of planted clique vertices.\\
\end{definition}

\begin{definition}[Clique Leakage $k$-Partite PC Detection Problem: ${\sf{clkPC_D}}(\ell \cdot k, k)$]\label{def:clkpc-detection}\ \\
This is the hypothesis testing problem defined analogous to ${\sf{kPC_D}}(\ell \cdot k, k)$ using ${\sf clkG}(\ell \cdot k, 1/2,k)$ instead of ${\sf kG}(\ell \cdot k, 1/2,k)$.\\
\end{definition}

\begin{conjecture}[Clique Leakage Logspace $k$-Partite PC Conjecture: {\sf clkPC-Conj-Space}]\label{conj:logspace-clkpc}\ \\
This is the conjecture analogous to {\sf kPC-Conj-Space} involving ${\sf{clkPC_D}}(\ell \cdot k, k)$ instead of  ${\sf{kPC_D}}(\ell \cdot k, k)$.\\
\end{conjecture}

\begin{lemma}[Multiple access randomness from the Clique Leakage $k$-Partite PC input]\label{lem:clkpc-randomness-harvesting}\ \\
Let $\ell = \omega(1)$, $k$ and $k_s$ be any non-decreasing positive integer sequences such that $k = O(\ell^{1-\delta})$ for some constant $\delta > 0$ and $k_s \leq k$. Given $\ell \cdot k$ and $k$, assume we can deterministically compute and store $k_s$ using an additional $O(\log \ell)$ bits of space. 

If the {\sf clkPC-Conj-Space} (Conjecture~\ref{conj:logspace-clkpc}) is true, no multiple access randomized logspace algorithm using at most ${\ell \cdot k \choose 2}$ multiple access random bits can solve ${\sf{clkPC_D}}(\ell \cdot k_s,k_s)$.
\end{lemma}
\begin{proof}
Let $\mathcal{S}$ be any multiple access randomized logspace algorithm that uses at most ${\ell \cdot k \choose 2}$ multiple access random bits when run on an instance of ${\sf{clkPC_D}}(\ell \cdot k_s,k_s)$. We will show that $\mathcal{S}$ fails to solve ${\sf{clkPC_D}}(\ell \cdot k_s,k_s)$. 

Let $\mathcal{S'}$ be an algorithm which takes as input $\{\ell \cdot 4k,4k,A_G\}$ corresponding to an instance of ${\sf{clkPC_D}}(\ell \cdot 4k,4k)$. Let $A_s$ denote the adjacency matrix of the subgraph induced by $A_G$ (the adjacency matrix input to $\mathcal{S'}$) on its first $\ell \cdot k_s$ vertices. Let $A_r$ denote the adjacency matrix of the subgraph induced by $A_G$ on the following vertex subset. This vertex subset consists of those of the last $\ell \cdot 3k$ vertices of $A_G$ which \emph{do not have an edge} to its first vertex. Because $k_s \leq k$, there are no common vertices shared by $A_s$ and $A_r$. Let $r_G$ be the binary string containing the edge indicator bits for possible edges in the graph corresponding to $A_r$. Let $\mathcal{S'}$ run $\mathcal{S}$ on $A_s$ using $r_G$ as its source of at most ${\ell \cdot k \choose 2}$ multiple access random bits\footnote{If $A_r$ has fewer than $\ell \cdot k$ vertices, $\mathcal{S'}$ can output any arbitrary fixed answer.} and output the answer given by $\mathcal{S}$.

\begin{enumerate}
\item If the input graph to $\mathcal{S'}$ has the null distribution $A_G \sim \G(\ell \cdot 4k,1/2)$, the input graph to $\mathcal{S}$ is an instance of $A_s \sim \G(\ell \cdot k_s,1/2)$ with $r_G$ a string of independent uniform random bits.
\item If the input graph to $\mathcal{S'}$ has the planted distribution $A_G \sim {\sf clkG}(\ell \cdot 4k,1/2,4k)$, the input graph to $\mathcal{S}$ is an instance of $A_s \sim {\sf clkG}(\ell \cdot k_s,1/2,k_s)$ with $r_G$ a string of independent uniform random bits that are independent of all other randomness. The former is because the first vertex of $A_G$ is also the first vertex of $A_s$. The latter is by construction since none of the vertices in $A_r$ can be in the planted clique vertex set of $A_G$.
\end{enumerate}

\begin{itemize}[leftmargin=0pt]
\item \textbf{$\mathcal{S'}$ is a randomized logspace algorithm:} Using $O(\log \ell)$ space, $\mathcal{S'}$ can count how many vertices there are in $A_r$ as well as provide multiple access to $A_r$ by using its ability to check if any given vertex has an edge to the first vertex in $A_G$. The space complexity analysis of $\mathcal{S'}$ now proceeds like in the proof of Lemma~\ref{lem:basic-randomness-harvesting}. We omit the details. Combined with the {\sf clkPC-Conj-Space} (Conjecture~\ref{conj:logspace-clkpc}) and the fact that $4k = O(\ell^{1-\delta})$, this means $\mathcal{S'}$ fails to solve ${\sf{clkPC_D}}(\ell \cdot 4k,4k)$.
\item \textbf{$\mathcal{S}$ fails to solve ${\sf{clkPC_D}}(\ell \cdot k_s,k_s)$:} Let ${\sf EnoughVertices}$ denote the event that there are at least $\ell \cdot k$ vertices in the subgraph corresponding to the adjacency matrix $A_r$. Standard tail bounds for binomial random variables imply that $\Pr({\sf EnoughVertices}) = 1-o(1)$ irrespective of whether $A_G \sim \G(\ell \cdot 4k,1/2)$ or $A_G \sim {\sf clkG}(\ell \cdot 4k,1/2,4k)$. Conditioned on ${\sf EnoughVertices}$, the random string $r_G$ has enough random bits for $\mathcal{S}$'s multiple access needs, and the probability that $\mathcal{S'}$ outputs the correct answer is equal to the probability that $\mathcal{S}$ outputs the correct answer\footnote{Crucially, the randomness that determines whether ${\sf EnoughVertices}$ occurs is independent of the randomness in $A_s$ and only affects the size of $r_G$, not its independence and uniformity.}. Since $\mathcal{S'}$ fails and $\Pr({\sf EnoughVertices}) = 1-o(1)$, $\mathcal{S}$ must fail to solve ${\sf clkPC_D}(\ell \cdot k_s,k_s)$.\qedhere
\end{itemize}
\end{proof}

\subsection{Clique Leakage Hypergraph Planted Clique}
\label{subsec:rand-from-clhpc}

\begin{definition}[Erd\H{o}s-R\'enyi hypergraph distribution: ${\sf HG}^s(n, 1/2)$]\label{def:h-er}\ \\
Let $G = ([n],E)$ be a $s$-uniform hypergraph with vertex set of size $n$. The hyperedge set $E$ is created by including each possible hyperedge independently with probability $\frac{1}{2}$. The distribution on $s$-uniform hypergraphs thus formed is denoted ${\sf HG}^s(n, 1/2)$.\\
\end{definition}

\begin{definition}[Hypergraph Planted Clique distribution: ${\sf HG}^s(n, 1/2,k)$]\label{def:hpc}\ \\
Let $G = ([n],E)$ be a $s$-uniform hypergraph with vertex set of size $n$. Let $K \subset [n]$ be a set of size $k$ chosen uniformly at random from all ${n \choose k}$ subsets of size $k$. For all sets of $s$ distinct vertices $u_1,u_2,...,u_s \in K$, we add the hyperedge $(u_1,u_2,...,u_s)$ to $E$. 
	For all other possible hyperedges, we add the hyperedge to $E$ independently with probability $\frac{1}{2}$. The distribution on $s$-uniform hypergraphs thus formed is denoted ${\sf HG}^s(n, 1/2,k)$.\\
\end{definition}

\begin{definition}[Clique Leakage Hypergraph Planted Clique distribution: ${\sf clHG}^s(n, 1/2,k)$]\label{def:clhpc}\ \\
This is the hypergraph planted clique distribution ${\sf HG}^s(n, 1/2,k)$ conditioned on the event that the first $s-2$ vertices are in the planted clique vertex set $K$.\\
\end{definition}

\begin{definition}[Clique Leakage Hypergraph PC Detection Problem: ${\sf clHPC}_{\sf D}^s(n,k)$]\label{def:clhpc-detection}\ \\
Given non-decreasing positive integer sequences $n=\omega(1)$ and $k \leq n$, this is the hypothesis testing problem induced by \[{\sf H}_0(n) = \left\{{\sf HG}^s(n,1/2)\right\} \hspace{0.1cm} \text{ and } \hspace{0.1cm} {\sf H}_1(n) = \left\{{\sf clHG}^s(n, 1/2,k)\right\}\] where the input to problem is the quadruple given by the number of vertices $n$, the size of the planted clique $k$, the hyperedge set of the random hypergraph, and $s$.\\
\end{definition}

\begin{conjecture}[Clique Leakage Logspace Hypergraph PC Conjecture: {\sf clHPC-Conj-Space}]\label{conj:logspace-clhpc}\ \\
Let $\A$ be any randomized logspace algorithm whose read-once uniformly random bit string is denoted $r_{ro}$. Let $s \geq 3$ be any constant integer and $n=\omega(1)$, $k \leq n$ be any sequences of non-decreasing positive integers such that $k = O(n^{\frac{1}{2}-\delta})$ for some constant $\delta > 0$. Then \[
	\underset{\hspace{-1em}\mathrlap{A_G \sim {\sf HG}^s(n,1/2),~r_{ro}}}{\Prob}\left( \A(\{n,k,A_G,s\},r_{ro}) = 0 \right) +
	\underset{\hspace{-1em}\mathrlap{A_G \sim {\sf clHG}^s(n,1/2,k),~r_{ro}}}{\Prob}\left(\A(\{n,k,A_G,s\},r_{ro}) = 1\right) = 1+o(1).\]\\
\end{conjecture}

\begin{lemma}[Arbitrarily large polynomial amounts of multiple access randomness from the Clique Leakage Hypergraph Planted Clique input]\label{lem:clhpc-randomness-harvesting}\ \\
Let $n=\omega(1)$ and $k$ be any non-decreasing positive integer sequences such that $k=O(n^{\frac{1}{2}-\delta})$ for some positive constant $\delta > 0$.

If {\sf clHPC-Conj-Space} (Conjecture~\ref{conj:logspace-clhpc}) is true, no multiple access randomized logspace algorithm can solve the planted clique detection problem ${\sf PC_D}(n,k)$.
\end{lemma}
\begin{proof}
Our proof will strongly mirror that of Lemma~\ref{lem:basic-randomness-harvesting}.
Let $\mathcal{S}$ be any multiple access randomized logspace algorithm. Since such algorithms can only use $\poly(n)$ amounts of randomness, there must exist a constant positive integer $s \geq 3$ such that for all large enough $n$, the number of multiple access random bits $\mathcal{S}$ uses when run on the input corresponding to ${\sf{PC_D}}(n,k)$ is at most ${n+s - 2 \choose s-1}$. We will show that $\mathcal{S}$ fails to solve ${\sf{PC_D}}(n,k)$.

Let $\mathcal{S'}$ be an algorithm which takes as input the quadruple $\{n+s-1,k+s-2,A_G,s\}$ corresponding to an instance of ${\sf clHPC}_{\sf D}^s(n+s-1,k+s-2)$. Here $A_G$ encodes the hyperedge set of the $s$-uniform hypergraph on $n+s-1$ vertices. Let $A_{\sf PC}$ denote the adjacency matrix of an $n$-vertex graph whose $n$ vertices correspond to the vertex subset $V_{\sf PC} = \{s-1,s,...,n+s-2\}$ of the hypergraph $A_G$. The edge set of $A_{\sf PC}$ depends on the hyperedge set of $A_G$ as follows. For any two distinct vertices $i < j \in V_{\sf PC}$, there is an edge in $A_{\sf PC}$ if and only if the size $s$ hyperedge $(1,2,...,s-2,i,j)$ exists in $A_G$. Let $r_G$ denote the string of size ${n+s - 2 \choose s-1}$ containing the hyperedge indicators for all possible  hyperedges in $A_G$ involving the last vertex (i.e the vertex named $n+s-1$) in $A_G$. Let $\mathcal{S'}$ run $\mathcal{S}$ on $A_{\sf PC}$ using $r_G$ as its source of multiple access random bits\footnote{If $\mathcal{S}$ needs more multiple access random bits than the size of $r_G$ (as is possible if $n$ is not large enough) $\mathcal{S'}$ can output any arbitrary fixed answer.} and output the answer given by $\mathcal{S}$.

\begin{itemize}[leftmargin=0pt]
\item \textbf{$\mathcal{S'}$ is a randomized logspace algorithm:} The space complexity analysis of $\mathcal{S'}$ is essentially like in the proof of Lemma~\ref{lem:basic-randomness-harvesting}. We omit the details. Combined with the {\sf clHPC-Conj-Space} (Conjecture~\ref{conj:logspace-clhpc}) and the fact that $k+s-2 = O((n+s-1)^{\frac{1}{2}-\delta})$, this means $\mathcal{S'}$ fails to solve ${\sf clHPC}_{\sf D}^s(n+s-1,k+s-2)$.
\item \textbf{$\mathcal{S}$ fails to solve ${\sf{PC_D}}(n,k)$:} Let ${\sf LastVertexNotInClique}$ denote the event that the $(n+s-1)^{th}$ vertex in $A_G$ is not in the planted clique set when $A_G \sim {\sf clHG}^s(n+s-1, 1/2,k+s-2)$. Observe that $\Pr({\sf LastVertexNotInClique}) =1-o(1)$. Conditioned on {\sf LastVertexNotInClique}, it is clear that $\mathcal{S}$ (when invoked by $\mathcal{S'}$) is being run on an instance of ${\sf{PC_D}}(n,k)$. This is because 
\begin{enumerate}[leftmargin=0pt]
\item Irrespective of whether the input hypergraph is ${\sf HG}^s(n+s-1,1/2)$ or ${\sf clHG}^s(n+s-1,1/2,k+s-2)$, $r_G$ consists of iid uniform random bits that are independent of $A_{\sf PC}$.
\item If $A_G \sim {\sf HG}^s(n+s-1,1/2)$, then $A_{\sf PC} \sim \G(n,1/2)$. On the other hand, if $A_G \sim {\sf clHG}^s(n+s-1,1/2,k+s-2)$, then $A_{\sf PC} \sim \G(n,1/2,k)$.
\end{enumerate}
Conditioned on {\sf LastVerticesNotInClique}, the probability that $\mathcal{S'}$ outputs the correct answer is equal to the probability that $\mathcal{S}$ outputs the correct answer when run on a instance of ${\sf PC_D}(n,k)$. Since $\mathcal{S'}$ fails and $\Pr({\sf LastVerticesNotInClique}) = 1-o(1)$, $\mathcal{S}$ must fail to solve ${\sf PC_D}(n,k)$.\qedhere
\end{itemize}
\end{proof}


\section{Evidence for Planted Clique Hardness despite Clique Leakage}
\label{sec:clique-leakage-hardness-evidence}

\begin{definition}[Hypergraph PC Detection Problem: ${\sf HPC}_{\sf D}^s(n,k)$]\label{def:hpc-detection}\ \\
This is the hypothesis testing problem analogous to ${\sf clHPC}_{\sf D}^s(n,k)$ that uses ${\sf HG}^s(n, 1/2,k)$ instead of ${\sf clHG}^s(n, 1/2,k)$.\\
\end{definition}

\begin{conjecture}[Logspace Hypergraph PC Conjecture: {\sf HPC-Conj-Space}]\label{conj:logspace-hpc}\ \\
This is the conjecture analogous to {\sf clHPC-Conj-Space} involving ${\sf HPC}_{\sf D}^s(n,k)$ instead of  ${\sf clHPC}_{\sf D}^s(n,k)$.\\
\end{conjecture}

\begin{definition}[Deterministic Logspace Algorithm for Planted Clique Recovery]\label{defn:recovery-planted-clique-variants}\ \\
Given a `planted clique input' (i.e a random object distributed as any of the variants of the planted clique distribution ${\sf G}(n,1/2,k)$ / ${\sf kG}(\ell \cdot k, 1/2,k)$ / ${\sf clkG}(\ell \cdot k, 1/2,k)$ / ${\sf HG}^s(n,1/2,k)$ / ${\sf clHG}^s(n,1/2,k)$), a deterministic logspace algorithm $\A$ solves the planted clique recovery task if it outputs the true planted clique $K$ with at least constant probability over the randomness in the input. That is, if $G$ denotes the planted clique input\footnote{Recall that this input also contains problem parameters such as $n$ (or $\ell \cdot k$), $k$, and $s$.}, the following happens with at least constant probability. For any ${\sf vertex} \in [n]$ (or ${\sf vertex} \in [\ell \cdot k]$), $\A(\{G,{\sf vertex}\})$ is $1$ if ${\sf vertex} \in K$ and $0$ if ${\sf vertex} \notin K$.\\
\end{definition}

\begin{lemma}[Hardness of Clique Leakage Hypergraph PC Recovery from {\sf HPC-Conj-Space}]\label{lem:clhpc-recovery-hardness}\ \\
Let $n = \omega(1)$ and $k$ be non-decreasing positive integer sequences such that $\omega((\log n)^{\frac{1}{s-1}}) = k = O(n^{\frac{1}{2}-\delta})$ for some constants $s \geq 3$ and $\delta>0$.

If {\sf HPC-Conj-Space} is true, no deterministic logspace algorithm can solve the clique leakage hypergraph planted clique recovery problem (Definition~\ref{defn:recovery-planted-clique-variants} with ${\sf clHG}^s(n,1/2,k)$).
\end{lemma}
\begin{proof}
Suppose a deterministic logspace algorithm $\A$ can solve the planted clique recovery problem with ${\sf clHG}^s(n,1/2,k)$. We construct a deterministic logspace algorithm that solves the ${\sf HPC}_{\sf D}^s(n,k)$, yielding a contradiction.

Let $\mathcal{S}$ be an algorithm that gets as input $\{n,k,A_G,s\}$ corresponding to the input of ${\sf HPC}_{\sf D}^s(n,k)$. $\mathcal{S}$ iterates over all $n^{s-2}$ tuples $(v_1,v_2,...,v_{s-2})$ where each $v_i \in [n]$. It can do this efficiently in lexicographic order by maintaining and incrementing $s-2$ counters named $v_1,v_2,...,v_{s-2}$ that use $O(\log n)$ bits each. For a given tuple, it checks if $v_1 < v_2 < ... < v_{s-2}$. This can be done using $O(\log n)$ bits of space. If this condition is not met, $\mathcal{S}$ moves on to the next tuple\footnote{Essentially, this is a space efficient way to iterate over all vertex sets (rather than tuples) of size $s-2$.}. If this condition \textit{is} met, however, $\mathcal{S}$ does the following.

Letting $t$ denote the current tuple $(v_1,...,v_{s-2})$, we define a $s$-uniform hypergraph $A_t$ that is essentially $A_G$ but with the vertices renamed. Let ${\sf rename}_t: [n] \rightarrow [n]$ be a permutation we will define soon. $A_t$ has a hyperedge $(i_1,i_2,...,i_s)$ if an only if $A_G$ has a hyperedge $({\sf rename}_t(i_1),{\sf rename}_t(i_2),...,{\sf rename}_t(i_s))$. For $i \in [s-2]$, we define ${\sf rename}_t(i) := v_i$. For $i \in [n] \setminus [s-2]$ we define ${\sf rename}_t(i)$ to be the $(i-(s-2))^{th}$ smallest positive integer not in $\{v_1,v_2,...,v_{s-2}\}$. Observe that ${\sf rename}_t$ is a permutation and can be computed using $O(\log n)$ bits of space for any input $i \in [n]$. Hence we can check the existence of any given $s$-uniform hyperedge in $A_t$ using $O(\log n)$ bits of space. Crucially, if the $\{v_1,...,v_{s-2}\}$ corresponding to $t$ are all planted clique vertices in $A_G$, the first $s-2$ vertices are planted clique vertices in $A_t$.

$\mathcal{S}$ now uses the deterministic logspace algorithm $\A$ on the renamed hypergraph $A_t$ and counts the number of vertices in $A_t$ for which $\A$ outputs a $1$ (i.e the number of vertices in $A_t$ that $\A$ believes are in a planted clique). If this number is less than $k$, $\mathcal{S}$ moves on to the next tuple. If this number is at least $k$, $\mathcal{S}$ does the following. If the vertices in $A_t$ for which $\A$ outputs $1$ form a clique\footnote{This can be checked with $O(\log n)$ bits of working space by iterating over all size $s$ subsets of $[n]$ (using the scheme described earlier in the proof). For every set of $s$ vertices $\{i_1,...,i_s\}$ such that $\A$ outputs $1$ on all of them, we check for a hyperedge $(i_1,...,i_s)$ in $A_t$.}, $\mathcal{S}$ terminates and outputs $1$ (indicating its belief that the input hypergraph has a planted clique). If not, $\mathcal{S}$ moves on to the next tuple. If $\mathcal{S}$ iterates over all tuples without terminating as described above, it terminates and outputs a $0$ (indicating its belief that the input hypergraph does not have a planted clique).

First, it is clear by construction that $\mathcal{S}$ is a deterministic logspace algorithm. We now argue that $\mathcal{S}$ solves ${\sf HPC}_{\sf D}^s(n,k)$. It is impossible for $\mathcal{S}$ to output $1$ if there is no clique of size $k$ in $A_G$. If $A_G \sim {\sf HG}^s(n, 1/2)$, it is folklore (Lemma~\ref{lem:Clique-size-in-Hyp-ER}) that $A_G$ can not have a clique of size $k$ except with probability $o(1)$. Hence the probability that $\mathcal{S}$ outputs the correct answer is at least $1-o(1)$.

If, on the other hand, $A_G \sim {\sf HG}^s(n, 1/2,k)$, there must be at least one tuple $t = (v_1,v_2,...,v_{s-2})$ such that $A_t \sim {\sf clHG}^s(n, 1/2,k)$. By assumption, with at least constant probability, $\A$ will correctly identify the planted clique vertices in $A_t$ and $\mathcal{S}$ will output the correct answer $1$. This gives \[
\underset{\hspace{-1em}\mathrlap{A_G \sim {\sf HG}^s(n,1/2)}}{\Prob}\left( \mathcal{S}(\{n,k,A_G,s\}) = 0 \right) +
\underset{\hspace{-1em}\mathrlap{A_G \sim {\sf HG}^s(n,1/2,k)}}{\Prob}\left(\mathcal{S}(\{n,k,A_G,s\}) = 1\right) = 1+\Omega(1)\] and completes the proof.
\end{proof}

\begin{lemma}[Hardness of Clique Leakage $k$-Partite PC Recovery from {\sf kPC-Conj-Space}]\label{lem:clkpc-recovery-hardness}\ \\
	Let $\ell = \omega(1)$ and $k$ be non-decreasing positive integer sequences such that $\omega(\log (\ell \cdot k)) = k = O(\ell^{1-\delta})$ for some constant $\delta>0$.
	
	If {\sf kPC-Conj-Space} is true, no deterministic logspace algorithm can solve the clique leakage $k$-Partite planted clique recovery problem (Definition~\ref{defn:recovery-planted-clique-variants} with ${\sf clkG}(\ell \cdot k,1/2,k)$).
\end{lemma}
\begin{proof}
The proof is analogous to that of Lemma~\ref{lem:clhpc-recovery-hardness} so we omit the details. The only difference is that we now iterate over all vertices $i \in [\ell]$ rather than tuples of size $s-2$ as in Lemma~\ref{lem:clhpc-recovery-hardness}. For any vertex $i \in [\ell]$, the renamed graph simply swaps vertices $1$ and $i$ and leaves the names of all other vertices unchanged.
\end{proof}


\section{Logspace Reductions using auxiliary Multiple Access Randomness}
\label{sec:reductions-to-other-probs}

\subsection{Sparse Principal Component Analysis \cite{berthet2013complexity}}
\label{subsec:spca}

Setting up some notation from \cite{berthet2013complexity}, define $R_0:=\left\{(\bar{d},\bar{n},\bar{k}) \in \mathbb{N}^3 : 15\sqrt{\frac{\bar{k}\log (120e\bar{d})}{\bar{n}}} \leq 1, \bar{k} \leq \bar{d}^{0.49} \right\}$\footnote{The $120$ is arbitrary and can be any constant greater than $18$. Similarly the $0.49$ can be any constant less than $0.5$.} and $R_{\mu}:=R_0 \cap \left\{\bar{n}^{\mu} \leq \bar{k} \right\} \cap \left\{ \bar{n} < \bar{d}  \right\}$ for any $\mu \in (0,1)$.

The Sparse PCA Hypothesis Testing Problem (${\sf SPCA}(\bar{d},\bar{n},\bar{k},\bar{\theta})$) \cite{berthet2013complexity} is a problem where the input consists of the parameters $\bar{d}$, $\bar{n}$, $\bar{k}$, $\bar{\theta}$ (which must lie in $R_{\mu}$) as well as $\bar{n}$ independent samples $X_1,...,X_{\bar{n}} \in \mathbb{R}^{\bar{d}}$ of a $\bar{d}$-dimensional random vector $X$. The null hypothesis  is the set of all distributions ($\mathcal{D}_0$) whose empirical variance along any direction is `small'. The alternate hypothesis is the set of all distributions ($\mathcal{D}_1^{\bar{k}}(\bar{\theta})$) for which there is a $\bar{k}$-sparse direction along which the empirical variance is `large', as quantified by $\bar{\theta}$. Since we will not need the detailed definitions of $\mathcal{D}_0$ and $\mathcal{D}_1^{\bar{k}}(\bar{\theta})$, we do not provide them here and point the interested reader to \cite{berthet2013complexity}.\\

\begin{lemma}[\cite{berthet2013complexity}'s reduction to Sparse PCA can be implemented in mutiple access randomized logspace]\label{lem: spca-br13-redn}\ \\
Let $\delta \in \left(0,\frac{1}{6} \right)$ and $\mu \in \left[\frac{1}{3},\frac{1}{2}-\delta \right)$ be constants. Let $n=\omega(1)$, $k = \Theta(n^{\frac{1}{2}-\delta})$, and $m=o(n/k)$ be non-decreasing positive integer sequences with the following properties. Given $n$ and $k$, the value $m$ can be computed and stored using $O(\log n)$ additional bits of space and $n-m$ is always a multiple of $2$. Let $\frac{n-m}{2} < \bar{d} \leq \poly(n)$, $\bar{n} = \left\lceil \left(\frac{4(n-m)}{k}\right)^{\frac{1}{1-\mu}} \right\rceil = o(n)$, and $ \bar{n}^{\mu}\leq \bar{k} = \left\lceil \frac{\bar{n}k}{4(n-m)} \right\rceil \leq \bar{d}^{0.49}$ be non-decreasing positive integer sequences. In particular, $(\bar{d},\bar{n},\bar{k}) \in R_{\mu}$. Define $\bar{\theta} := \frac{(\bar{k}-1)k}{n-m} =  \Theta\left(\sqrt{\frac{\bar{k}^{4-\frac{1}{\mu}}}{\bar{n}}}\right)$. Assume that given $n-m$ and $k$, we can deterministically compute and store $\bar{d}$ using $O(\log n)$ additional bits of space.

If there is a randomized logspace algorithm that solves ${\sf SPCA}(\bar{d},\bar{n},\bar{k},\bar{\theta})$, there is a multiple access randomized logspace algorithm that solves ${\sf PC_D}(n-m,k)$ using at most $\left(\bar{d}-\left(\frac{n-m}{2}\right)\right)\cdot \bar{n} + 2 \cdot 10 \bar{d} (\lceil \log \bar{d} \rceil)^2 + \bar{n} \leq \poly(n-m)$ multiple access random bits.
\end{lemma}
\begin{proof}
The proof follows easily from observing that the reduction of \cite{berthet2013complexity} can be (approximately) implemented space efficiently. Let $\mathcal{S}$ be a randomized logspace algorithm that solves ${\sf SPCA}(\bar{d},\bar{n},\bar{k},\bar{\theta})$. Let $\mathcal{S'}$ be a multiple access randomized algorithm that has $\left(\bar{d}-\left(\frac{n-m}{2}\right)\right)\cdot \bar{n} + 2 \cdot 10 \bar{d} (\lceil \log \bar{d} \rceil)^2 + \bar{n}$ multiple access random bits and $\{n-m,k,A_G \}$ as input corresponding to ${\sf PC_D}(n-m,k)$. Clearly, $\mathcal{S'}$ can compute and store $\bar{d}$, $\bar{n}$, $\bar{k}$, and $\bar{\theta}$ (the latter can be represented by storing its numerator and denominator) using $O(\log n)$ bits of space.

Let $\bar{A}$ denote a $\bar{d} \times \bar{n}$ matrix whose first $\frac{n-m}{2}$ rows are simply the `bottom-left' submatrix of $A_G$. That is, for $i \in [\frac{n-m}{2}]$, the $i^{th}$ row of $\bar{A}$ is simply the $\left(\frac{n-m}{2}+i\right)^{th}$ row of $A_G$ restricted to its first $\bar{n}$ columns\footnote{In \cite{berthet2013complexity}, a uniformly random permutation is applied to the matrix $A_G$ before extracting this `bottom-left' submatrix. However, in our formulation of the planted clique problem, all sets of size $k$ are equally likely to be the planted vertex set, so $A_G$'s distribution (in both the null and planted cases) is invariant to permutations. Hence we can skip this step.}. The rest of the $\bar{d}-\left(\frac{n-m}{2}\right)$ rows of $\bar{A}$ are populated by independent uniform random bits. Let $\pi_{\bar{d}} : [\bar{d}] \rightarrow [\bar{d}]$ and $\pi_{\bar{n}} : [\bar{n}] \rightarrow [\bar{n}]$ be independent random efficiently computable functions (which are approximately uniform random permutations) sampled by invoking Lemma~\ref{lem:permutation} with $10 \bar{d} (\lceil \log \bar{d} \rceil)^2$ and $10 \bar{n} (\lceil \log \bar{n} \rceil)^2$ independent uniform multiple access random bits respectively.

Let $\bar{B}$ denote the $\bar{d} \times \bar{n}$ matrix whose $(i,j)^{th}$ entry is the $(\pi_{\bar{d}}(i),\pi_{\bar{n}}(j))^{th}$ entry of $\bar{A}$ and $\bar{B}_j$ denote the $j^{th}$ column of $\bar{B}$. For $j \in [\bar{n}]$, define $X_j = \eta_j(2\bar{B}_j-1) \in \{-1,1\}^{\bar{d}}$ where $\eta_j$ is an independent uniform random bit. Because the $\pi$'s can be computed space efficiently (see Lemma~\ref{lem:permutation}), $\mathcal{S'}$ can use its $\left(\bar{d}-\left(\frac{n-m}{2}\right)\right)\cdot \bar{n} + 2 \cdot 10 \bar{d} (\lceil \log \bar{d} \rceil)^2 + \bar{n}$ multiple access random bits to provide multiple access to  the vectors $X_1,...,X_{\bar{n}}$. Since $\bar{d}\bar{n} \leq \poly(n)$, as a multiple access randomized logspace algorithm, $\mathcal{S'}$ can invoke $\mathcal{S}$ on $X_1,...,X_{\bar{n}}$ and output the corresponding answer.

We know from Lemma~\ref{lem:permutation} and the tensorization inequality for total variation distance \cite[fact 3.1]{brennan2019optimal} that the pair $(\pi_{\bar{d}},\pi_{\bar{n}})$ we use has vanishing $o(1)$ total variation distance from a pair of independent uniformly random permutations. The data processing inequality then implies that our space efficient algorithm $\mathcal{S'}$ solves ${\sf PC_D}(n-m,k)$ if and only if an analogous algorithm $\mathcal{S''}$ solves ${\sf PC_D}(n-m,k)$. $\mathcal{S''}$ uses a pair of independent uniformly random permutations instead of $(\pi_{\bar{d}},\pi_{\bar{n}})$. Crucially, we do not care about the space complexity of $\mathcal{S''}$ and can invoke results from \cite{berthet2013complexity}, where it was studied. For the rest of the proof, assume the $\pi$'s are independent uniformly random permutations.

As noted in the proof of \cite[Theorem 7]{berthet2013complexity}, if $A_G \sim \G(n-m,1/2)$, then the vectors $X_1,...,X_{\bar{n}}$ are drawn from a distribution in the null hypothesis of ${\sf SPCA}(\bar{d},\bar{n},\bar{k},\bar{\theta})$. On the other hand, if $A_G \sim {\sf G}(n-m,1/2,k)$, the collection of vectors $X_1,...,X_{\bar{n}}$ has total variation distance at most $\frac{16\bar{n}}{n-m}=o(1)$ from a collection of $\bar{n}$ independent samples from some distribution in the alternate hypothesis of ${\sf SPCA}(\bar{d},\bar{n},\bar{k},\bar{\theta})$. This last fact follows from \cite[Lemma 8]{berthet2013complexity} because $\frac{2\bar{n}k}{n-m}\geq \max\{16\log n,8\bar{k}\}$\footnote{\cite[Lemma 8]{berthet2013complexity} actually states a worse upper bound on the total variation distance. However, a close look at their proof shows that the better total variation distance upper bound we state also holds.}. By the data processing inequality, the probability that $\mathcal{S''}$ outputs the correct answer on an instance of ${\sf PC_D}(n-m,k)$ is within $o(1)$ of the probability that $\mathcal{S}$ outputs the correct answer on an instance of ${\sf SPCA}(\bar{d},\bar{n},\bar{k},\bar{\theta})$. This complete the proof.
\end{proof}

\begin{theorem}[Logspace hardness of Sparse PCA (Restricted parameter range)]\label{thm:spca-restricted}\ \\
Let $\bar{d}$, $\bar{n}$, $\bar{k}$, and $\bar{\theta}$ be as defined in Lemma~\ref{lem: spca-br13-redn} with the additional constraints that $m=\omega(log^2 n)$ and $\bar{d}  =\left(\frac{n-m}{2}\right)+ o(\frac{nm}{\bar{n}})$. If the logspace planted clique conjecture {\sf PC-Conj-Space} (Conjecture~\ref{conj:logspace-pc}) is true, no randomized logspace algorithm can solve ${\sf SPCA}(\bar{d},\bar{n},\bar{k},\bar{\theta})$.
\end{theorem}
\begin{proof}
We have $\left(\bar{d}-\left(\frac{n-m}{2}\right)\right)\cdot \bar{n} + 2 \cdot 10 \bar{d} (\lceil \log \bar{d} \rceil)^2 + \bar{n} = o(mn) \leq {n \choose 2} - {n-m \choose 2}$ for large enough $n$. Since $m$ can be computed and stored using $O(\log n)$ bits of space given $n$ and $k$, we can obtain the desired conclusion by combining Lemma~\ref{lem: spca-br13-redn} with Lemma~\ref{lem:basic-randomness-harvesting}.
\end{proof}
Theorem~\ref{thm:spca-restricted} shows that we can use Conjecture~\ref{conj:logspace-pc} to deduce the randomized logspace hardness of the Sparse PCA problem as long as $\bar{d}$ is not too small and not too large compared to $\bar{n}$. However, we believe the Sparse PCA problem is hard even for much larger values of $\bar{d}$. Theorem~\ref{thm:spca-full}, which follows immediately by combining Lemma~\ref{lem: spca-br13-redn} with Lemma~\ref{lem:clhpc-randomness-harvesting}, can show randomized logspace hardness for Sparse PCA instances with such parameters.\\
\begin{theorem}[Logspace hardness of Sparse PCA]\label{thm:spca-full}\ \\	Let $\bar{d}$, $\bar{n}$, $\bar{k}$, and $\bar{\theta}$ be as defined in Lemma~\ref{lem: spca-br13-redn}. If the clique leakage logspace hypergraph planted clique conjecture {\sf clHPC-Conj-Space} (Conjecture~\ref{conj:logspace-clhpc}) is true, no randomized logspace algorithm can solve ${\sf SPCA}(\bar{d},\bar{n},\bar{k},\bar{\theta})$.
\end{theorem}

\subsection{Planted Submatrix Detection \cite{ma2015computational}}
\label{subsec:submat}

Following \cite{ma2015computational}, for any real number $x$ and integer $t$, denote the $t$-bit truncation of $x$ as $[x]_t := 2^{-t}\lfloor 2^t x \rfloor$. We extend this notation to matrices of real numbers by applying the truncation entrywise.\\

\begin{definition}[Planted Gaussian Matrix: ${\sf PGM}(\bar{p},\bar{\theta},\bar{t})$]\label{def:gauss-mat}\ \\
This is the distribution on $\bar{p} \times \bar{p}$ matrices $X = (X_{ij})$ whose entries are independent truncated Gaussians $[\mathcal{N}(\bar{\theta}_{ij},1)]_{\bar{t}}$. Here $\bar{\theta} = (\bar{\theta}_{ij})$ is \textit{any} $\bar{p} \times \bar{p}$ matrix with real entries.\\
\end{definition}

\begin{definition}[Planted Submatrix Detection: ${\sf Submat}(\bar{p},\bar{k},\bar{\lambda},\bar{t})$ \cite{ma2015computational}]\label{def:submat-detection}\ \\
Let $\bar{p} = \omega(1)$, $\bar{k} \leq \bar{p}$, and $\bar{t}$ be non-decreasing positive integer sequences. Let $\bar{\lambda}$ be a sequence of positive numbers. This is the hypothesis testing problem (with $\bar{p}$, $\bar{k}$, $[\bar{\lambda}]_{\bar{t}}$, and $\bar{t}$ part of the input.) induced by ${\sf H}_0(\bar{p}) = {\sf PGM}(\bar{p},0_{\bar{p}},\bar{t})$\footnote{$0_{\bar{p}}$ denotes the $\bar{p} \times \bar{p}$ all zeros matrix.} and \[{\sf H}_1(\bar{p}) = \left\{ {\sf PGM}(\bar{p},\bar{\theta},\bar{t}): \exists U,V 
\subset [\bar{p}] \text{ such that } \lvert U \rvert, \lvert V \rvert \geq \bar{k}, \bar{\theta}_{ij} \geq \bar{\lambda} \text{ if } (i,j) \in U \times V, \bar{\theta}_{ij}=0 \text{ otherwise}  \right\}.\]
\end{definition}

\begin{lemma}[\cite{ma2015computational}'s reduction from planted clique to ${\sf Submat}(\bar{p},\bar{k},\bar{\lambda},\bar{t})$ can be implemented in multiple access randomized logspace]\label{lem:submat-redn}\ \\
Let $\ell = \omega(1)$, $k_s$, $\bar{p}$, $\omega(1) = \bar{k} \leq \bar{p}$, and $\bar{t} = \lceil 4\log \bar{p} \rceil$ be non-decreasing sequences of positive integers. Let $\bar{\lambda}$ be a sequence of positive numbers. Suppose $20\bar{k} = k_s =  O(\ell^{1-\delta})$ for some constant $\delta >0$. Further, $\bar{p}$ is a factor of $\frac{\ell \cdot k_s}{2}$ such that $ 2k_s \leq \bar{p} \leq \frac{\ell \cdot k_s}{2}$ and $\bar{\lambda} \leq \frac{\bar{p}}{(\ell \cdot k_s)\sqrt{6 \log (\ell \cdot k_s)}}$. Assume that $k_s$, $\bar{p}$, and $[\bar{\lambda}]_{\bar{t}}$ can be computed and stored given $\ell$ using $O(\log \ell)$ bits of space.

If there is a randomized logspace algorithm that solves ${\sf Submat}(\bar{p},\bar{k},\bar{\lambda},\bar{t})$, there is a multiple access randomized logspace algorithm that solves ${\sf kPC_D}(\ell \cdot k_s,k_s)$ (or ${\sf clkPC_D}(\ell \cdot k_s,k_s)$) using at most $O((\ell\cdot k_s)^2 \log \ell)$ multiple access random bits.
\end{lemma}
\begin{proof}
Because of the reduction in Lemma~\ref{lem:kpc-to-pc}, it suffices to demonstrate a multiple access randomized logspace algorithm that solves ${\sf PC_D}(\ell \cdot k_s,k_s)$ using at most $O((\ell\cdot k_s)^2 \log \ell)$ multiple access random bits. We show that the reduction from \cite[Theorem 4]{ma2015computational} can yield such an algorithm since $\bar{k} = \omega(1)$.

Our parameters are slightly more constrained than the corresponding parameters in \cite{ma2015computational}. However, since we have only introduced (rather than removed) constraints, the reduction's proof of correctness from \cite{ma2015computational} carries over unchanged, and we only need to reason about its space complexity. These extra constraints are extremely mild and technical, and do not affect the scaling of any parameters involved. They ensure that the number of vertices ($\ell \cdot k_s$) is an integral multiple of the number of planted clique vertices ($k_s$) because we want to reduce from $k$-partite versions of the clique problem. Further, they ensure $k_s$, $\bar{p}$, and $[\bar{\lambda}]_{\bar{t}}$ can be stored and computed given $\ell$ using $O(\log \ell)$ bits of space. The first of these is useful in downstream applications when we need to invoke our randomness harvesting lemmas from Sections~\ref{subsec:rand-from-kpc} and \ref{subsec:rand-from-clkpc}. The latter are useful in our current proof to ensure our algorithm `knows' the size of the ${\sf Submat}$ instance it should reduce to\footnote{\cite{ma2015computational} does not need these technical conditions because they use a (very weakly) non-uniform hardness assumption.}.

Let $\mathcal{S}$ be a randomized logspace algorithm that solves ${\sf Submat}(\bar{p},\bar{k},\bar{\lambda},\bar{t})$. Let $\mathcal{S'}$ be a multiple access randomized algorithm that has $\Theta((\ell \cdot k_s)^2 \log \ell)$ multiple access random bits and $\{\ell \cdot k_s,k_s,A_G \}$ as input corresponding to ${\sf PC_D}(\ell \cdot k_s,k_s)$. $\mathcal{S'}$ can compute and store $\bar{p}$ and $\bar{k}$ (and hence $\bar{t}$ and $[\bar{\lambda}]_{\bar{t}}$) using $O(\log \ell)$ bits of space. Let $\bar{M}:=\sqrt{6\log (\ell \cdot k_s)}$, $\bar{\mu}:=\frac{1}{2\bar{M}}$, $\bar{w}:=\bar{t}+6\lceil\log(\ell \cdot k_s)\rceil$, and $\bar{T}:=\lceil \log \bar{M}\rceil + \bar{w} + 3 \lceil \log (\ell \cdot k_s)\rceil$. Let $Q_0$ and $Q_1$ be two distributions (to be defined later) on the support of $[-\bar{M},\bar{M}]_{\bar{w}}$. Any sample from either distribution can be described using $\lceil \log \bar{M} \rceil + \bar{w} = O(\log \ell)$ bits. Further assume (an assumption we will justify in the final paragraph of this proof) that sampling an independent sample from either of them can be done using $O(\log \ell)$ bits of space and $\bar{T}$ independent uniform random bits.

Define the integer $N_2 := \frac{\ell \cdot k_s}{2}$ and let $B_0$ (respectively $B_1$) be a $N_2 \times N_2$ matrix filled with independent sample from $Q_0$ (respectively $Q_1$). $\mathcal{S'}$ can provide multiple  access to entries of $B_0$ and $B_1$ since it has multiple access to $2 \cdot  N_2^2 \cdot \bar{T} = O((\ell \cdot k_s)^2 \log \ell)$ independent uniform random bits. Letting $A \in \{0,1\}^{N_2 \times N_2}$ be the lower left quarter of the input adjacency matrix $A_G$, define $B$ as the $N_2 \times N_2$ matrix whose $(i,j)^{th}$ entry is ${B_0}_{ij}(1-{A_0}_{ij})+{B_1}_{ij}({A_0}_{ij})$. Noting that $\bar{p}$ divides $N_2$ by assumption, partition $B$ into $(\frac{N_2}{\bar{p}})^2$ consecutive blocks of size $\bar{p} \times \bar{p}$ and let $X$ denote the $\bar{p} \times \bar{p}$ matrix that is the sum of all blocks divided by $\frac{N_2}{\bar{p}}$. Clearly $\mathcal{S'}$ can provide multiple access to entries of $X$. $\mathcal{S'}$ now outputs the answer given by $\mathcal{S}$ when invoked on a $[X]_{\bar{t}}$ along with the corresponding parameters. Since $\bar{p} \leq \poly(\ell)$, $\mathcal{S'}$ can simulate this call to $\mathcal{S}$.

To complete the proof, we justify that the distributions $Q_0$ and $Q_1$ used by \cite{ma2015computational} can indeed be sampled from using $\bar{T}$ independent uniform random bits and $O(\log \ell)$ bits of space. Remark 5 in \cite{ma2015computational} describes how the `inverse CDF' technique can be used to do this if we can compute any desired entry in the CDF (equivalently pmf) of either distribution. This technique can easily be implemented space efficiently using linear search because the support of each distribution has at most $\bar{M}2^{\bar{w}+1} \leq \poly(\ell)$ elements. Hence we only need that the probability of any sample of either distribution can be computed using $O(\log \ell)$ bits of space. These values are all, by construction, representable using $\bar{T}$ bits of space. Further, they are $\bar{T}$ bit truncations of real numbers defined using simple univariate operations such as addition, subtraction, multiplication, division, and integration of smooth functions, which can all be computed up to $s$ bits of accuracy using $O(s)$ bits of space. Computing all these operations using some large enough constant multiple of $\bar{T}$ bits of precision (which is still $O(\log \ell)$) will give us a number whose first $\bar{T}$ bits are indeed the correct pmf value. Hence each of these values can be computed using $O(\log \ell)$ bits of space.
\end{proof}

\begin{theorem}[Logspace hardness of planted submatrix detection (Restricted parameter range)]\label{thm:submat-restricted}\ \\
Let $\bar{p}$, $\bar{k}$, $\bar{\lambda}$ and $\bar{t}$ be as described in Lemma~\ref{lem:submat-redn} with the additional constraint that $\bar{k} = O(\ell^{\frac{1}{2}-\delta})$ for some constant $\delta > 0$. If the logspace $k$-partite planted clique conjecture {\sf kPC-Conj-Space} (Conjecture~\ref{conj:logspace-kpc}) is true, no randomized logspace algorithm can solve ${\sf Submat}(\bar{p},\bar{k},\bar{\lambda},\bar{t})$.
\end{theorem}
\begin{proof}
Let $k = m = \lceil \ell^{1-\frac{\delta}{2}} \rceil$ be non-decreasing positive integer sequences. $k_s = 20\bar{k} = O(\ell^{\frac{1}{2}-\delta})$, which gives $O((\ell\cdot k_s)^2 \log \ell) = o(\ell \cdot k \cdot m) \leq {\ell \cdot k \choose 2} - {\ell \cdot k-m \choose 2}$ for large enough $\ell$. Combining Lemma~\ref{lem:basic-kpc-randomness-harvesting} (using $k_s$, $k$, and $m$ as defined) with Lemma~\ref{lem:submat-redn} and using the fact that $k_s$ and $m$ are computable from $\ell$ using $O(\log \ell)$ bits of space completes the proof.
\end{proof}
An example illustrates the impact of the restriction $\bar{k} = O(\ell^{\frac{1}{2}-\delta})$ required in Theorem~\ref{thm:submat-restricted}. Consider parameter sequences such that $\bar{\lambda} = \Theta(1 / \sqrt{\log \bar{p}})$ and $\bar{k} = \Theta(\bar{p}^{\frac{1}{3}})$. Even though we believe that a submatrix detection instance with such parameters can not be solved by randomized logspace algorithms, we can not use Theorem~\ref{thm:submat-restricted} to show this. We can, however, use Theorem~\ref{thm:submat-full} to circumvent this restriction.\\

\begin{theorem}[Logspace hardness of planted submatrix detection]\label{thm:submat-full}\ \\
Let $\bar{p}$, $\bar{k}$, $\bar{\lambda}$ and $\bar{t}$ be as described in Lemma~\ref{lem:submat-redn}. If the clique leakage logspace $k$-partite planted clique conjecture {\sf clkPC-Conj-Space} (Conjecture~\ref{conj:logspace-clkpc}) is true, no randomized logspace algorithm can solve ${\sf Submat}(\bar{p},\bar{k},\bar{\lambda},\bar{t})$.
\end{theorem}
\begin{proof}
Let $k = \lceil \ell^{1-\frac{\delta}{2}} \rceil$ be a non-decreasing positive integer sequence. Since $k_s = O(\ell^{1-\delta})$, we have $O((\ell\cdot k_s)^2 \log \ell) \leq {\ell \cdot k \choose 2}$ for large enough $\ell$. Combining Lemma~\ref{lem:clkpc-randomness-harvesting} (using $k_s$ and $k$ as defined) with Lemma~\ref{lem:submat-redn} and using the fact that $k_s$ is computable from $\ell$ using $O(\log \ell)$ bits of space completes the proof.
\end{proof}

\subsection{Testing Almost $k$-wise Independence \cite{alon2007testing}}
\label{subsec:k-wise}

\begin{definition}[$(\bar{\epsilon},\bar{k})$-wise Independence: Definition 2.1 in \cite{alon2007testing}]\label{def:eps-k-wise-indep}\ \\
A distribution $D$ supported on $\{0,1\}^{\bar{n}}$ is $(\bar{\epsilon},\bar{k})$-wise independent if for any $\bar{k}$ indices $i_1<i_2<...<i_{\bar{k}}$ and any vector $v \in \{0,1\}^{\bar{k}}$, $\lvert \underset{x \sim D}{\Prob}\left(x_{i_1}...x_{i_{\bar{k}}} = v\right)-2^{-{\bar{k}}}\rvert \leq \bar{\epsilon}$. Let $\mathcal{D}_{(\bar{\epsilon},\bar{k})}$ denote the set of all $(\bar{\epsilon},\bar{k})$-wise independent distributions.\\
\end{definition}

\begin{definition}[Testing $(\bar{\epsilon},\bar{k})$-wise vs $(\bar{\epsilon'},\bar{k})$-wise Independence: ${\sf Indep}(\bar{n},\bar{k},\bar{\epsilon},\bar{\epsilon'},\bar{s})$]\label{def:eps-vs-eps-prime}\ \\
Let $\bar{n}=\omega(1)$, $\bar{k} \leq \bar{n}$, and $\bar{s} \leq \poly(\bar{n})$ be non-decreasing positive integer sequences. Let $0 < \bar{\epsilon} < \bar{\epsilon'} < 1$ be sequences. The input consists of all the aforementioned parameters and $\bar{s}$ independent samples from a distribution $D$ supported on $\{0,1\}^{\bar{n}}$.
This is the hypothesis testing problem induced by \[{\sf H}_0(\bar{n}) = \left\{D: D \in \mathcal{D}_{(\bar{\epsilon},\bar{k})} \right\} \hspace{0.1cm} \text{ and } \hspace{0.1cm} {\sf H}_1(\bar{n}) = \left\{D: D \notin \mathcal{D}_{(\bar{\epsilon'},\bar{k})}\right\}.\]
\end{definition}

\begin{lemma}[\cite{alon2007testing}'s reduction from ${\sf PC_D}(n,k)$ to ${\sf Indep}(\bar{n},\bar{k},\bar{\epsilon},\bar{\epsilon'},\bar{s})$ can be implemented in multiple access randomized logspace]\label{lem:k-wise-redn}\ \\
Let $0 < \alpha \leq 1$ and $\delta > 0$ be any constants and let $\ell = \omega(1)$, $n=2^{\ell}$, $\omega(\log^2 n) = k = O(n^{\frac{1}{2}-\delta})$, $\bar{k} = 1+\ell$, $\bar{n} = 2^{\frac{\bar{k}}{\alpha}}$, and $\bar{s} \leq \poly(\bar{n})$ be non-decreasing positive integer sequences. Define sequences $0 < \bar{\epsilon} < \bar{\epsilon'} < 1$ as $\bar{\epsilon} := \frac{2\alpha \log^2(\bar{n})}{\bar{n}^{\alpha}}$ and $\bar{\epsilon'} := \frac{k-2}{\bar{n}^{\alpha}}$. Assume that $\bar{s}$ can be computed and stored using $O(\log n)$ bits of space given $n$ and $k$. 

If there is a randomized logspace algorithm that solves ${\sf Indep}(\bar{n},\bar{k},\bar{\epsilon},\bar{\epsilon'},\bar{s})$, there is a multiple access randomized logspace algorithm that solves ${\sf PC_D}(n,k)$.
\end{lemma}
\begin{proof}
The proof is a straightforward observation that the reduction of \cite{alon2007testing} can be implented using a multiple access randomized logspace algorithm. We only reason about the space complexity of this reduction. Its correctness follows immediately from the analysis in \cite{alon2007testing}. Let $\mathcal{S}$ be any randomized logspace algorithm that solves ${\sf Indep}(\bar{n},\bar{k},\bar{\epsilon},\bar{\epsilon'},\bar{s})$.

Let $\mathcal{S'}$ be a multiple access randomized logspace algorithm that receives as input an instance $\{n,k,A_G\}$ of ${\sf PC_D}(n,k)$. It can store and compute all the parameters $\bar{n},\bar{k},\bar{\epsilon},\bar{\epsilon'},\bar{s}$ using $O(\log n)$ bits of space. Consider a matrix $B \in \{0,1\}^{n \times \bar{n}}$ whose first $n$ columns are the adjacency matrix $A_G$ augmented with independent uniform random bits on the diagonal. The last $\bar{n}-n$ columns of $B$ consists entirely of more independent uniform random bits. Because $\mathcal{S'}$ has multiple access to $\poly(n)$ independent uniform random bits (and $\bar{n} \leq \poly(n)$), it can simulate multiple access to $B$. Consider a further $\ell \cdot \bar{s}$ multiple access random bits. $\mathcal{S'}$ can use them to simulate multiple access to $\bar{s}$ independent uniformly random rows of $B$. $\mathcal{S'}$ interprets these rows as $\bar{s}$ samples from a distribution on $\{0,1\}^{\bar{n}}$, invokes $\mathcal{S}$ on them, and outputs $\mathcal{S}$'s answer. $\mathcal{S'}$ can do this because $\bar{s} \cdot \bar{n} \leq \poly(n)$, so invoking $\mathcal{S}$ takes $O(\log n)$ bits of space. \cite{alon2007testing} shows that $\mathcal{S'}$ solves ${\sf PC_D}(n,k)$. \qedhere
\end{proof}

Combining Lemma~\ref{lem:k-wise-redn} with Lemma~\ref{lem:clhpc-randomness-harvesting} immediately yields the following.

\begin{theorem}[Logspace hardness of testing almost $k$-wise independence]\label{thm:k-wise}\ \\
Let $\bar{n}$, $\bar{k}$, $\bar{\epsilon}$, $\bar{\epsilon'}$, and $\bar{s}$ be as defined in Lemma~\ref{lem:k-wise-redn}. If the clique leakage logspace hypergraph planted clique conjecture {\sf clHPC-Conj-Space} (Conjecture~\ref{conj:logspace-clhpc}) is true, no randomized logspace algorithm can solve ${\sf Indep}(\bar{n},\bar{k},\bar{\epsilon},\bar{\epsilon'},\bar{s})$.
\end{theorem}

\cite{alon2007testing} also consider a different formalization of the testing problem. However, they show that this other problem is equivalent to ${\sf Indep}(\bar{n},\bar{k},\bar{\epsilon},\bar{\epsilon'},\bar{s})$. That is, the two problems are related by the identity reduction. This is why we have only focused on ${\sf Indep}(\bar{n},\bar{k},\bar{\epsilon},\bar{\epsilon'},\bar{s})$.

\section{Auxiliary Lemmas}
\label{sec:auxlem}

\begin{lemma}[Space efficient uniformly random permutation from uniform random bits]\label{lem:permutation}\ \\
Let $n=\omega(1)$ and $r$ be non-decreasing positive integer sequences. Suppose $ 10 n (\lceil \log n \rceil)^2 \leq r \leq \poly(n)$. Given $r$ independent uniform random bits, there is a random function $\pi : [n] \rightarrow [n]$ with the following properties.
\begin{itemize}[leftmargin=0pt]
\item  Given multiple access to the $r$ random bits, $\pi(i)$ can be computed using a further $O(\log n)$ bits of space when given \textit{any} $i \in [n]$.
\item The total variation distance between $\pi$ and a uniformly random permutation from $[n] \rightarrow [n]$ decays as $O\left(n \cdot  \exp\left(\frac{-r}{2n\lceil \log n \rceil}\right)\right)$.
\end{itemize}

\end{lemma}
\begin{proof}
Divide the $r$ random bits into contiguous chunks of size $\lceil \log n \rceil$\footnote{Formally, the $j^{th}$ chunk consists of the bits $[j \cdot \lceil \log n \rceil] \setminus [(j-1) \cdot \lceil \log n \rceil]$.}. Each chunk represents an integer in $[2^{\lceil \log n \rceil}]$. Define $\pi : [n] \rightarrow [n]$ as follows. Enumerating over chunks in the natural ordering, let $\pi(i)$ be the $i^{th}$ \textit{distinct} integer in $[n]$ that we see represented by the random bits in a chunk. If there are fewer than $i$ distinct integers in $[n]
$ represented in all the chunks combined, let $\pi(i)$ be any fixed integer in $[n]$.

To compute $\pi(i)$, we maintain three $O(\log n)$ bit counters. Using the first counter, we enumerate over chunks in the natural ordering. We increment the second counter every time we see a chunk whose random bits represent an integer in $[n]$ we have not seen so far. We can do this as follows. Every time the first counter sees an integer in $[n]$, we use the third counter to enumerate over all previous chunks to check if we have already seen that integer. Once the second counter reaches $i$, we output the integer represented in the corresponding chunk as $\pi(i)$. If the first counter has enumerated over all chunks and the second counter does not reach $i$, output any fixed value. This shows the space complexity of computing $\pi$ is as claimed.

Standard tail bounds for the coupon collector problem imply that except with probability at most $O\left(n \cdot  \exp\left(\frac{-r}{2n\lceil \log n \rceil}\right)\right)$, each integer $i \in [2^{\lceil \log n \rceil}]$ (and hence in $[n]$) is represented in at least one chunk. Conditioned on every integer in $[n]$ being represented in at least one chunk, $\pi$ clearly defines a uniform random permutation. This proves the total variation bound using the well-known ``conditioning on an event" property of total variation distance \cite[Fact 3.1]{brennan2019optimal}.\qedhere\\
\end{proof}

\begin{lemma}[Reduction from $k$-Partite versions of Planted Clique to ${\sf PC_D}$]\label{lem:kpc-to-pc}\ \\
Let $\ell = \omega(1)$, $r$, and $k$ be non-decreasing positive integer sequences such that $r \geq 10 \ell \cdot k (\lceil \log (\ell \cdot k) \rceil)^2$. Assume there is a multiple access randomized logspace algorithm that solves ${\sf PC_D}(\ell \cdot k,k)$ using at most $r - 10 \ell \cdot k (\lceil \log (\ell \cdot k) \rceil)^2$ multiple access random bits. Then there is a multiple access randomized logspace algorithm that solves ${\sf kPC_D}(\ell \cdot k,k)$ or ${\sf clkPC_D}(\ell \cdot k,k)$ using $r$ multiple access random bits.
\end{lemma}
\begin{proof}
Use the first $10 \ell \cdot k (\lceil \log (\ell \cdot k) \rceil)^2$ multiple access random bits to sample an efficiently computable random function  $\pi: [\ell \cdot k] \rightarrow [\ell \cdot k]$ using Lemma~\ref{lem:permutation}. Given an instance of ${\sf kPC_D}(\ell \cdot k,k)$ or ${\sf clkPC_D}(\ell \cdot k,k)$, invoke the algorithm for ${\sf PC_D}(\ell \cdot k,k)$ on the following graph and output the answer it gives. Given two vertices $i,j$, there is an edge between them if and only if there is an edge between the vertices $\pi(i)$ and $\pi(j)$ in the input graph. Clearly this construction gives a randomized logspace algorithm using $r$ multiple access random bits.
If $\pi$ were a uniformly random permutation, this new graph would be distributed as ${\sf PC_D}(\ell \cdot k,k)$. Since the total variation distance between $\pi$ and a uniformly random permutation is $o(1)$, the data processing inequality for total variation distance \cite[Fact 3.1]{brennan2019optimal} implies that the constructed algorithm solves our problem.\qedhere
\end{proof}

The following fact is folklore.
\begin{lemma}[Maximum clique size in an Erd\H{o}s-R\'enyi hypergraph ${\sf HG}^s(n, 1/2)$]\label{lem:Clique-size-in-Hyp-ER}\ \\
For $n=\omega(1)$, with probability at least $1-o(1)$, the size of the largest clique in an Erd\H{o}s-R\'enyi hypergraph ${\sf HG}^s(n, 1/2)$ is $O((\log n)^{\frac{1}{s-1}})$.
\end{lemma}
\begin{proof}
This follows from a simple union bound argument. The probability of any fixed vertex subset of size $t \geq s$ being a clique is $2^{-{t \choose s}}$. Since there are ${n \choose t}$ vertex subsets of size $t$, a union bound implies that the probability of there existing a clique of size $t$ is upper bounded by ${n \choose t} \cdot 2^{-{t \choose s}}$. If $t = \omega((\log n)^{\frac{1}{s-1}})$, this upper bound is $o(1)$.
\end{proof}

\section*{Acknowledgments}
We would like to thank Yanjun Han, Ray Li, and Greg Valiant for helpful discussions and feedback that improved the presentation of these results.

\addcontentsline{toc}{section}{References}
\bibliographystyle{alpha}
\bibliography{ref}
	

\end{document}